\documentclass{lmcs} 
\pdfoutput=1

\usepackage{lastpage}
\lmcsdoi{17}{2}{20}
\lmcsheading{}{\pageref{LastPage}}{}{}%
{Jan.~27,~2020}{May~27,~2021}{}

\usepackage[utf8]{inputenc}

\keywords{well-quasi-ordering, recursive path ordering, iterative path ordering, Hydra battle, termination, Kruskal's tree theorem, Buchholz Hydra}

\usepackage{hyperref}

\usepackage{latexsym}
\usepackage{amsmath,amssymb,amstext,amsthm}
\usepackage{wasysym}

\usepackage{graphicx}
\usepackage{xcolor}
\usepackage{subcaption}
\usepackage{stmaryrd}

\usepackage{tikz}
\usetikzlibrary{arrows,arrows.meta,automata,backgrounds,calc,chains,decorations.fractals,decorations.pathreplacing,fadings,fit,folding,mindmap,patterns,plotmarks,positioning,shadows,shapes.geometric,shapes.symbols,through,trees}

\colorlet{dblue}{blue!40!black}

\usepackage{textcomp}

\usepackage{enumitem}
\usepackage{verbatim}
\usepackage{float}

\usepackage{stackengine}
\usepackage{scalerel}
\newcommand\openbigstar[1][0.7]{%
  \raisebox{.5ex}{\scalerel*{%
    \stackinset{c}{-.125pt}{c}{}{\scalebox{#1}{\color{white}{$\bigstar$}}}{%
      $\bigstar$}%
  }{\bigstar}}
}

\theoremstyle{plain}
\newtheorem{notation}[thm]{Notation}
\newtheorem{definition}[thm]{Definition}
\newtheorem{theorem}[thm]{Theorem}
\newtheorem{proposition}[thm]{Proposition}
\newtheorem{lemma}[thm]{Lemma}
\newtheorem{remark}[thm]{Remark}
\newtheorem{example}[thm]{Example}

\newtheorem{corollary}[thm]{Corollary}

\def\setminus{-}

\newcommand{\mrm}{\mathrm}

\newcommand{\mcl}{\mathcal}

\newcommand{\mit}{\mathit}

\renewcommand{\prop}[1]{\ensuremath{\mrm{#1}}}
\newcommand{\SN}{\prop{SN}}
\newcommand{\SNr}{\funap{\SN}}

\newcommand{\CR}{\prop{CR}}
\newcommand{\CRr}{\funap{\CR}}

\newcommand{\funap}[2]{#1(#2)}
\newcommand{\bfunap}[3]{\funap{#1}{#2,#3}}

\renewcommand{\emptyset}{\varnothing}

\newcommand{\nat}{\mathbb{N}}

\newcommand{\sred}{{\rightarrow}}

\newcommand{\smred}{{\twoheadrightarrow}}
\newcommand{\mred}{\mathrel{\smred}}
\newcommand{\smredi}{{\twoheadleftarrow}}
\newcommand{\mredi}{\mathrel{\smredi}}

\newcommand{\posemp}{\epsilon}

\newcommand{\pos}{\funap{\mathcal{P}\!os}}

\newcommand{\redrat}[2]{\mathrel{\sred_{#1,#2}}}

\newcommand{\pairlft}{(}
\newcommand{\pairrgt}{)}
\newcommand{\pairsep}{{,\,}}
\newcommand{\pairstr}[1]{\pairlft#1\pairrgt}
\newcommand{\pair}[2]{\pairstr{#1\pairsep#2}}
\newcommand{\triple}[3]{\pairstr{#1\pairsep#2\pairsep#3}}

\providecommand\rightarrowRHD{\relbar\joinrel\mathrel\RHD}
\providecommand\rightarrowrhd{\relbar\joinrel\mathrel\rhd}

\newcommand{\ilpored}{\rightarrowrhd}
\newcommand{\ilporedomega}{\rightarrowRHD}
\newcommand{\ilporedomegamulti}{\rightarrowRHD_{\#}}

\newcommand{\ilporule}[1]{\text{\sc{#1}}}
\newcommand{\lput}{\ilporule{put}}
\newcommand{\lselect}{\ilporule{select}}

\newcommand{\lcopy}{\ilporule{copy}}
\newcommand{\llex}{\ilporule{lex}}
\newcommand{\ldown}{\ilporule{down}}

\newcommand{\set}[1]{\{\,#1\,\}}
\newcommand{\mset}[1]{\{\!\{\,#1\,\}\!\}}
\newcommand{\quotient}[2]{#1/_{#2}}

\newcommand{\commute}{\simeq}
\newcommand{\commutestep}{\to_{\commute}}
\newcommand{\treecommute}[1]{#1_{\commute}}

\newcommand{\streeemb}{\preceq}
\newcommand{\treeemb}{\mathrel{\streeemb}}

\newcommand{\streenodes}{\mathrm{nodes}}
\newcommand{\treenodes}{\funap{\streenodes}}

\newcommand{\streelabel}{\mathrm{label}}
\newcommand{\treelabel}{\funap{\streelabel}}

\newcommand{\stre}{\mathbb{T}}
\newcommand{\tre}{\bfunap{\stre}}

\renewcommand{\star}{{\scalebox{.6}{\openbigstar[.5]\!\!\!}}}

\newcommand{\ster}{\mathcal{T}}
\newcommand{\ter}{\bfunap{\ster}}
\newcommand{\avars}{\mathcal{X}}
\newcommand{\vars}{\funap{\mcl{V}\!\mit{ar}}}
\newcommand{\atrs}{\mathcal{R}}
\newcommand{\aars}{\mathcal{A}}

\newcommand{\subtermat}[2]{#1|_{#2}}

\tikzset{arsdefault/.style={thick,shorten >= 0mm, shorten <= .5mm,inner sep=.3mm,outer sep=0mm}}
\tikzset{anno/.style={scale=.85,blue}}
\tikzset{wave/.style={decorate,decoration={snake,amplitude=.4mm,segment length=2mm,post length=1mm}}}
\tikzset{bluered/.style={cblue,nodes={black}}}
\tikzset{redred/.style={cred,nodes={black},wave}}
\tikzset{smallCircle/.style={circle,fill=black,inner sep=0mm,outer sep=1mm,minimum size=1mm}}

\newcommand{\fun}[1]{\mrm{#1}}

\tikzset{default/.style={
  thick,
  every node/.style={circle},
  level distance=12mm, 
  inner sep=.5mm}}
\tikzset{smallCircle/.style={circle,fill=black,inner sep=0mm,outer sep=1mm,minimum size=1mm}}
\tikzset{paint/.style={very thick,draw=#1!50!black,fill=#1,opacity=.4}}
\tikzset{paintopaque/.style={very thick,draw=#1!50!black!60,fill=#1!60}}
\tikzset{loopl/.style={out=140,in=210,looseness=10}}
\tikzset{loopr/.style={out=-40,in=30,looseness=10}}
\tikzset{loopb/.style={out=-40-90,in=30-90,looseness=10}}
\tikzset{loopt/.style={out=-40+90,in=30+90,looseness=10}}
\tikzset{loop/.style={out=-40+#1,in=30+#1,looseness=10}}

\tikzset{emptystep/.style={-,dotted,line cap=round,dash pattern=on 0 off 3.00000}}

\tikzset{b/.style={anchor=north,at=(#1.south)}}
\tikzset{br/.style={anchor=north west,at=(#1.south east)}}
\tikzset{bl/.style={anchor=north east,at=(#1.south west)}}
\tikzset{bw/.style={anchor=north west,at=(#1.south west)}}
\tikzset{be/.style={anchor=north east,at=(#1.south east)}}
\tikzset{a/.style={anchor=south,at=(#1.north)}}
\tikzset{ar/.style={anchor=south west,at=(#1.north east)}}
\tikzset{al/.style={anchor=south east,at=(#1.north west)}}
\tikzset{aw/.style={anchor=south west,at=(#1.north west)}}
\tikzset{ae/.style={anchor=south east,at=(#1.north east)}}
\tikzset{r/.style={anchor=west,at=(#1.east)}}
\tikzset{l/.style={anchor=east,at=(#1.west)}}
\tikzset{rn/.style={anchor=north west,at=(#1.north east)}}

\tikzset{eta/.style={very thick,->,cblue!90!black}}
\tikzset{beta/.style={very thick,->,corange!80!white!90!black}}

\tikzset{devcirc/.style={circle,draw,fill=white,inner sep=0,minimum size=4\pgflinewidth}}
\tikzset{dev/.style={postaction={decorate},decoration={
  markings,
  mark=at position .5 with \node [devcirc] {};}}
}
    
\tikzset{medium tree/.style={
    level 1/.style={sibling distance=17mm},
    level 2/.style={sibling distance=9mm},
    level 3/.style={sibling distance=5mm},
    level 4/.style={sibling distance=4mm},
  }}

\tikzset{startAt/.style={inner sep=0mm,r=#1,xshift=-1.2mm,yshift=.8mm}}

\newcommand{\arrowTriangle}[2]{
  \pgfpathmoveto{\pgfpoint{-0.01#1+#2}{.6#1}}
  \pgfpathlineto{\pgfpoint{#1+#2}{0}}
  \pgfpathlineto{\pgfpoint{-0.01#1+#2}{-.6#1}}
  \pgfusepathqfill
}

\newdimen\prearrowsize
\newdimen\arrowsize
\newdimen\temparrowsize
\newcommand{\arrowscale}{5}
\newcommand{\setarrowsize}{
  \arrowsize=0.000000001pt
  \prearrowsize=\arrowscale\pgflinewidth
  \normalizearrowsize
}
\newcommand{\normalizearrowsize}{
  \ifdim\prearrowsize>2mm
    \addtolength{\arrowsize}{2mm}
    \addtolength{\prearrowsize}{-2mm}
    \temparrowsize=0.5\prearrowsize
    \prearrowsize=\temparrowsize
  \else
  \fi

  \addtolength{\arrowsize}{\prearrowsize}
}

\pgfarrowsdeclarealias{share>}{share>}{stealth'}{stealth'}%

\pgfarrowsdeclare{red>}{red>}
{
  \setarrowsize
  \pgfarrowsleftextend{1.5\pgflinewidth} 
  \pgfarrowsrightextend{\arrowsize-0.1\arrowsize}
}{
  \setarrowsize
  \pgfsetdash{}{0pt} %
  \pgfsetroundjoin %
  \pgfsetroundcap %
  \arrowTriangle{\arrowsize}{-0.1\arrowsize}
}

\pgfarrowsdeclare{red>>}{red>>}
{
  \setarrowsize
  \pgfarrowsleftextend{1.5\pgflinewidth}
  \pgfarrowsrightextend{\arrowsize-0.1\arrowsize+.8\arrowsize}
}{
  \setarrowsize
  \pgfsetdash{}{0pt} %
  \pgfsetroundjoin %
  \pgfsetroundcap %
  \arrowTriangle{\arrowsize}{-0.1\arrowsize}
  \arrowTriangle{\arrowsize}{-0.1\arrowsize+.8\arrowsize}
}

\pgfarrowsdeclare{red>>>}{red>>>}
{
  \setarrowsize
  \pgfarrowsleftextend{1.5\pgflinewidth}
  \pgfarrowsrightextend{\arrowsize-0.1\arrowsize+1.6\arrowsize}
}{
  \setarrowsize
  \pgfsetdash{}{0pt} %
  \pgfsetroundjoin %
  \pgfsetroundcap %
  \arrowTriangle{\arrowsize}{-0.1\arrowsize}
  \arrowTriangle{\arrowsize}{-0.1\arrowsize+.8\arrowsize}
  \arrowTriangle{\arrowsize}{-0.1\arrowsize+1.6\arrowsize}
}

\tikzset{>=red>}
\pgfarrowsdeclarealias{<<}{>>}{red>>}{red>>}%
\pgfarrowsdeclarealias{<<<}{>>>}{red>>>}{red>>>}%

\tikzstyle{gyellow}=[draw=black!80,top color=white!50,bottom color=black!20]
\tikzstyle{gblue}=[draw=blue!50,top color=white,bottom color=blue!60]
\tikzstyle{gred}=[draw=red!50,top color=white,bottom color=red!60]
\tikzstyle{ggreen}=[draw=blue!80!green!90!black,top color=white,bottom color=blue!80!green!60]

\tikzstyle{roundNode}=[gyellow,thick,circle,minimum size=4mm,inner sep=0.5mm]

\definecolor{cblue}{rgb}{0,0.4,0.7}
\definecolor{clighterblue}{rgb}{0,0.6,1.0}
\colorlet{cred}{red}
\colorlet{cgreen}{green!80!black}
\colorlet{corange}{orange!70!red}
\colorlet{cpureorange}{orange}
\colorlet{cpurple}{clighterblue!50!cred}

\colorlet{clightblue}{clighterblue!50!cblue!40}%
\colorlet{clightred}{cred!40}
\colorlet{clightgreen}{cgreen!80!cblue!40}
\colorlet{clightyellow}{corange!40!yellow!50}
\colorlet{clightorange}{cred!50!orange!40}
\colorlet{clightpurple}{clighterblue!50!cred!50}

\colorlet{cdarkred}{red!70!black}
\colorlet{cdarkgreen}{green!60!black}

\colorlet{chighlight}{orange!50!yellow!60}

\tikzset{pgnode/.style={smallCircle,fill=white,draw=black,minimum size=1.3mm,outer sep=0.5mm}}
\tikzset{pgnodecolor/.style={pgnode,minimum size=4mm,scale=0.9}}
\tikzset{pgnodebig/.style={roundNode,gyellow,outer sep=1mm}}
\tikzset{pgrelation/.style={ultra thick,cblue!80!black,decorate,decoration={snake,amplitude=.4mm,segment length=4mm}}}

\tikzset{exi/.style={densely dotted}}
\tikzset{decweak/.style={-,green!50!black,very thick,opacity=0.7}}
\tikzset{decstrict/.style={->,orange,very thick,opacity=0.7}}

\newcommand{\contexthole}{\Box}
\newcommand{\cxthole}{\contexthole}
\newcommand{\cxtfill}[2]{#1[#2]}

\newcommand{\trsstar}{\mit{Star}}
\newcommand{\trsomega}{\mit{Star}^\omega}

\newcommand{\stepstarhydra}{\to_\mathit{SH}}

\newcommand{\hydrahead}[2]{
  \begin{pgfonlayer}{background}
  \begin{scope}[xscale=#2,yscale=-1,shift={(#1)}]
  \begin{scope}[shift={(-8mm,1mm)},scale=0.4]
  \draw [cblue!50,fill=cblue!50] (0,0) 
    to[out=80,in=180,looseness=1.3] ++(8mm,5mm)
    to[out=0,in=160,looseness=1.3] ++(12mm,3mm)
    to[out=160+180,in=-130,looseness=1] ++(12mm,3mm)
    to[out=-100,in=20,looseness=1] ++(-6mm,-7mm)
    to[out=160+180,in=-130,looseness=1] ++(8mm,2mm)
    to[out=-100,in=10,looseness=1] ++(-8mm,-6mm)
    to[out=160+180,in=-170,looseness=1] ++(8mm,0mm)
    to[out=-140,in=10,looseness=1] ++(-8mm,-5mm)
    to[out=10+180,in=20,looseness=1] ++(-25mm,0mm)
    to[out=20+180,in=-70,looseness=1] ++(-1mm,3mm)
    to[out=-20,in=-140,looseness=1] ++(8mm,1mm)
    to[out=140,in=0,looseness=1] cycle;
  
  \draw [fill=white,draw=none] (14mm,5mm) 
    to[out=20,in=180] ++(4mm,1.5mm)
    to[out=0,in=0] ++(0mm,-2mm)
    to[out=170,in=10] ++(-4mm,0mm)
    to[out=10+180,in=20+180] cycle;
  \end{scope}
  \end{scope}
  \end{pgfonlayer}
}

\begin{document}

\title{Star Games and Hydras}

\author[J\"{o}rg Endrullis]{J\"{o}rg Endrullis\rsuper{a}}
\address{
  \lsuper{a} Vrije Universiteit Amsterdam, Amsterdam
}
\email{j.endrullis@vu.nl, r.overbeek@vu.nl}

\author[Jan Willem Klop]{Jan Willem Klop\rsuper{b}}
\address{
  \lsuper{b}Centrum voor Wiskunde en Informatica (CWI), Amsterdam
}
\email{JWKlop1945@kpnmail.nl}

\author[Roy Overbeek]{Roy Overbeek\rsuper{a}}

\begin{abstract}
  The recursive path ordering is an established and crucial tool in term rewriting to prove termination. 
  We revisit its presentation by means of some simple rules on trees (or corresponding terms) equipped with a `star' as control symbol,
  signifying a command to make that tree (or term) smaller in the order being defined. 
  This leads to star games that are very convenient for proving termination of many rewriting tasks.
  For instance, using already the simplest star game on finite unlabeled trees, 
  we obtain a very direct proof of termination of the famous Hydra battle, direct in the sense that there is not the usual  mention of ordinals.
  We also include an alternative road to setting up the star games, 
  using a proof method of Buchholz, adapted by van Oostrom, resulting in a quantitative version of the star as control symbol.
  We conclude with a number of questions and future research directions.
\end{abstract}

\maketitle

\section*{Dedication}
\begin{quote}
  After graduating, Jos Baeten started his scientific work in the realm of \emph{the higher infinite} 
  with his Ph.D.\ thesis, appearing in the bibliography of Kanamori's book~\cite{kanamori1994higher} about that realm (see Baeten~\cite{baeten1984filters}). 
  The second author cherishes many precious memories of the cooperation with Jos and Jan Bergstra 
  in our years of process algebra around $1984$.

  The present note is dedicated to Jos in friendship, and includes several questions relating to the higher infinite of ordinals and cardinals.
\end{quote}

\section{Introduction}\label{sec:introduction}

To understand difficult notions in mathematics, logic or informatics it sometimes helps to develop a ``game view'' on the notion in question. 
A typical example is the notion of bisimulation. 
The second author of this note, at the start of his new job around Christmas 1980 at the CWI\footnote{%
  The Centrum Wiskunde \& Informatica of which Jos Baeten is the current director.
  In 1980 the CWI was still called Mathematical Centre Amsterdam.
}, in trying to understand the then newly proposed notion of the %
\emph{recursive path ordering} (RPO), found that his understanding was much eased by a ``gamification'', consisting of some simple rules for moving a control symbol $\star$ on finite labeled trees.

The star game version of RPO was called IPO, \emph{iterative path order}. IPO was used fruitfully by Bergstra \& Klop~\cite{bergstra1985algebra} to yield termination of a rewrite system evaluating process expressions in ACP, the \emph{algebra of communicating processes}. Together with Jan Bergstra and Aart Middeldorp, the second author of the current paper wrote a course book (in Dutch) Termherschrijfsystemen~\cite{bergstra1989termherschrijfsystemen}, while employed, respectively, in 1987 at the University of Amsterdam and the Vrije Universiteit Amsterdam. 
There the IPO star game is included, together with a proof of the crucial lemma of acyclicity and \emph{Kruskal's tree theorem}. 
This proof of the acyclicity lemma is published in the present paper for the first time (in English, the mentioned course book is in Dutch). 
The survey chapter~\cite{klop1993term} did not include a proof of the acyclicity lemma.

The iterative path order has been extended in various directions. 
Klop, van Oostrom and de Vrijer~\cite{klop2006iterative} have extended it to the \emph{iterative lexicographic path order} (ILPO) after including a well-known rule for lexicographic extensions of RPO\footnote{%
  For stars it was mentioned in~\cite{bergstra1989termherschrijfsystemen,klop1993term}. 
  For RPO it was defined in Dershowitz~\cite{dershowitz1991rewrite} and several introductions to RPO, e.g., Zantema's chapter on Termination in~\cite{zantema2003termination}.
}, see also \cite{gese:1990,baad:nipk:1998,dershowitz1982orderings,kamin2004two}.
In subsequent studies, van Raamsdonk and Kop~\cite{kop2008higher,kop2009iterative} have extended the framework of ILPO to higher-order rewriting~\cite{kop2008higher} and to deal with the recursive path order with \emph{status}~\cite{kop2009iterative}.
Nowadays, the recursive path order is a standard ingredient of every leading termination prover. In particular, it is an important part of the termination prover \textit{Jambox}~\cite{jambox,termination:automata:2015} developed by the first author. 
The recursive path order has also been applied fruitfully to prove confluence by decreasing diagrams~\cite{felg:oost:2013}.
Decreasing diagrams~\cite{oost:1994b,confluence:weak:diamond:2013,confluence:decreasing:diagrams:2018} is one of the strongest techniques for proving confluence of (abtract)  rewrite systems.

A peculiar matter with the IPO star game is that on the auxiliary objects, being finite trees adorned with control stars, the game reduction is \emph{not} terminating (\SN), while it is \SN{} when restricted to the proper, unstarred, objects.
Vincent van Oostrom~\cite{klop2006iterative} discovered a very interesting variant of the IPO star game where the stars are given a \emph{quantity of energy}, a natural number or ordinal. 
This variant had the remarkable property to be terminating even on the intermediate auxiliary trees (see Figure \ref{star-labels})
while inducing the same reduction relation on unstarred terms.
The termination proof in~\cite{klop2006iterative} did not use Kruskal's tree theorem plus acyclicity, but instead a very different proof method~\cite{buchholz1995proof}, developed by Buchholz, in the framework of proof theory investigations concerned with powerful independence theorems and fast-growing hierarchies. The method was adapted by van Oostrom in ~\cite{klop2006iterative}.
In our paper we have devoted a full section to this proof calling it `magical', which of course is a subjective appraisal, but which just means that we find this route surprisingly elegant and effective, saving us from the more laborious details of the symbol tracing proof of acyclicity.

\newcommand{\ilpopicbase}{
  \node (a1) [n] {};
  \node (a2) [n] at ($(a1) + (20:2cm)$) {};
  \node (a3) [n] at ($(a2) + (15:2cm)$) {};
  \node (a4) [n] at ($(a3) + (10:2cm)$) {};
  \node (r1) [n] at ($(a3) + (-40:1.8cm)$) {};
  \node (l1) [n] at ($(a2) + (70:2cm)$) {};
  \node (l2) [n] at ($(a3) + (60:1.5cm)$) {};
  
  \begin{scope}[->,very thick]
    \draw (a1) -- (a2);
    \draw (a2) -- (a3);
    \draw (a3) -- (a4);
    \draw (a3) -- (r1);
    \draw (r1) -- (a4);
    \draw (a2) -- (l1);
    \draw (a3) -- (l2);
  \end{scope}
}
\newcommand{\ilpopicextended}{
  \path (a1) to[bend left=-50] node [m,pos=.3] (t1) {} node [m,pos=.7] (t2) {} (a2);
  \draw (a1) to[bend left=-15] (t1); \draw (t1) to[bend left=-15] (t2); \draw (t2) to[bend left=-15] (a2);
  
  \node (ex1) [m] at ($(a2) + (120:.8cm)$) {};
  \node (ex2) [m] at ($(a2) + (120:1.6cm)$) {};
  \draw (a2) -- (ex1); \draw (ex1) -- (ex2); \draw (ex2) -- (l1);
  
  \node (loop) [m] at ($(a2) + (-20:1.3cm)$) {};
  \draw (a2) -- (loop); \draw (loop) -- (a3);

  \node (t1) [m] at ($(a3) + (100:1.0cm)$) {};
  \draw (a3) -- (t1); \draw (t1) -- (l2);

  \node (t1) [m] at ($(a3)!.5!(a4)!.4!(r1)$) {};
  \draw (a3) -- (t1); \draw (t1) -- (r1); \draw (t1) -- (a4);

  \node (infinite) [m] at ($(r1) + (30:1cm)$) {};
  \draw (r1) -- (infinite); \draw (infinite) -- (a4);
}
\newcommand{\ilpopicinfinite}{
  \node (loop2) [m] at ($(loop) + (3mm,-5mm)$) {};
  \node (loop3) [m] at ($(loop) + (-3mm,-5mm)$) {};
  \draw (loop) to[bend left=0] (loop2);
  \draw (loop2) to[bend left=0] (loop3);
  \draw (loop3) to[bend left=0] (loop);

  \node (ex3) [m] at ($(ex1) + (175:.8cm)$) {};
  \draw (ex1) -- (ex3); \draw (ex3) -- (ex2);

  \node (t2) [m] at ($(infinite) + (-10:1cm)$) {}; \draw (infinite) -- (t2);
  \node (t3) [m] at ($(t2) + (-20:.8cm)$) {}; \draw (t2) -- (t3);
  \node (t4) [rotate=-35] at ($(t3) + (-30:.4cm)$) {$\cdots$};
}

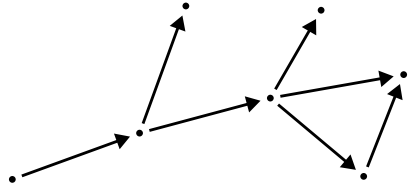
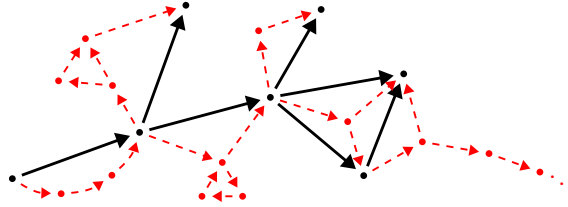
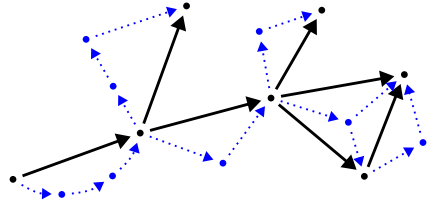
\begin{figure}[ht]
  \centering
  \begin{subfigure}{.8\textwidth}
    \centerline{\scalebox{0.9}{
      \begin{tikzpicture}[default,node distance=20mm,n/.style={smallCircle,fill=black}]
        \ilpopicbase
      \end{tikzpicture}
    }}
    \caption{Original terms with reduction relation $\to_1$, to be proved $\SN$.}\label{star-labels-a}
  \end{subfigure}\\[2ex]
  \begin{subfigure}{.8\textwidth}
    \centerline{\scalebox{0.9}{
      \begin{tikzpicture}[default,node distance=20mm,n/.style={smallCircle,fill=black}]
        \ilpopicbase
        \begin{scope}[->,red,m/.style={smallCircle,red}, dashed]
          \ilpopicextended
          \ilpopicinfinite
        \end{scope}
      \end{tikzpicture}
    }}
    \caption{Original terms together with auxiliary terms and reduction relation $\to_2$ 
            such that ${\to_1} \subseteq {\to_2^+}$, and $\to_2^+$ restricted to the original terms is $\SN$; $\to_2$ itself need not be $\SN$.
            It follows that $\to_1$ is $\SN$. This is the situation of ILPO \emph{with stars}.}\label{star-labels-b}
  \end{subfigure}\\[2ex]
  \begin{subfigure}{.8\textwidth}
    \centerline{\scalebox{0.9}{
      \begin{tikzpicture}[default,node distance=20mm,n/.style={smallCircle,fill=black}]
        \ilpopicbase
        \begin{scope}[->,blue,m/.style={smallCircle,blue}, dotted]
          \ilpopicextended
        \end{scope}
       \end{tikzpicture}
    }}
    \caption{Original terms together with auxiliary terms and reduction relation $\to_3$ 
            such that $\to_3$ is $\SN$, and ${\to_2^+} = {\to_3^+}$, 
            restricted to original terms. It follows that $\to_1$ is $\SN$. 
            This is the situation of ILPO \emph{with number labels}.}\label{star-labels-c}
  \end{subfigure}
  \caption{Rewriting of objects simulated by star rules (\subref{star-labels-b}), red and dashed, and by number star rules (\subref{star-labels-c}), blue and dotted. Note that the red auxiliary rewriting relation is not always terminating, while the blue one is. We note that the non-terminating red parts of the reduction are not simply cut away; their effect is in a more subtle way after finite unwinding subsumed in the blue terminating reduction.}
  \label{star-labels}
\end{figure}

Recently, in a draft book on abstract rewriting by the first two authors, we were including the Kirby-Paris Hydra battle~\cite{kirby1982accessible}, a paradigm example in the termination area~\cite{buchholz2007another,dershowitz2007hydra,moser2009hydra,touzet1998encoding}. The terminator Hercules succeeds in killing the monstrous Hydra, in the Greek mythology helped by the golden sword given to him by goddess Athena. While including this killer example of termination, a question, thirty years overdue, occurred to us: if the IPO star game is so flexible and versatile as it sometimes appeared to be, should it maybe also prove the victory of Hercules? And \emph{lo and behold}, the termination proof of the Kirby-Paris Hydra battle via a simulation in the IPO game, is almost immediately obvious, as we shall demonstrate in Section~\ref{sec:hydras}. Additionally, we design an even more monstrous Hydra, and show that it also is no match for IPO.

Remarkably, the usual proofs for termination of the Kirby-Paris Hydra battle use ordinals in the form of Cantor normal forms, or even the more intricate ordinal diagrams~\cite{takeuti2013proof}, but the proof via Kruskal's tree theorem and the star game in this paper does not refer to ordinals at all. Actually, ordinals are hiding in the background: the Kirby-Paris Hydra example requires only trees labelled with~$0$, i.e., unlabelled finite trees, and this game is paramount to the nested multiset ordering over natural numbers, which has, as is well-known, ordinal strength $\epsilon_0$~\cite{klop1993term}. The more general star game, employing natural numbers as node labels, is much stronger, reaching the ordinal $\phi_{\omega}(0)$ in the Veblen notation~\cite{dershowitz87}.

However, there are monsters and giant monsters, some of which are much harder to conquer than the Kirby-Paris Hydra.
The most well-known giant monster is the Buchholz Hydra~\cite{hamano1998direct,okada1988note,buchholz1987independence}, described in this paper. 
It is far beyond a simulation by the IPO star game as presented here, because its astronomically large ordinal dwarfs even the ordinals in the Veblen hierarchy.

\subsection*{Outline}
In the introduction we give a lengthy historical account of sources leading up to the current survey paper on star games as a termination framework.
In Section~\ref{sec:prelims} we recapitulate the main notions and notations for reading this paper, including a nutshell account of some of the first large ordinals.
Section~\ref{sec:star:games} presents the main star game of this paper,
the star game called \emph{iterative path order} (IPO).

In Section~\ref{sec:scenic:route} we establish termination of IPO using Kruskal's tree theorem and acyclicity. 
Section~\ref{sec:kruskal} gives a brief introduction to Kruskal's tree theorem and Higman's lemma.
In Section~\ref{sec:star:termination} we prove the pivotal lemma of acyclicity.

Section~\ref{sec:buchholz} contains a different route towards the same result of \SN{} of IPO, via a refinement of stars, and a proof method that is interesting in its own right.

Section~\ref{sec:hydras} presents three Hydras: the ``classical'' one by Kirby and Paris; a super Hydra of our making, still provably terminating by our IPO star game; and a hyper Hydra of Buchholz, that is beyond the power of IPO.

The final section, Section~\ref{sec:conclusion}, presents several directions of further research.

\section{Preliminaries}\label{sec:prelims}

In this section we will briefly enumerate and describe some of the main notions and notations figuring in this paper.
We will work with \emph{unranked} terms (basically labelled, ordered trees) where the symbols of the signature do \emph{not} have a fixed arity. 
We consider unordered trees (algebraic trees) as equivalence classes of terms modulo the permutation of arguments.

\subsection{Abstract rewriting}

\begin{definition}[Abstract rewriting]
  An \emph{abstract reduction system (ARS)} $\aars = \langle A, \to \rangle$ consists of a set $A$ of \emph{objects} and a binary \emph{rewrite} or \emph{reduction relation} ${\to} \subseteq {A \times A}$. 
\end{definition}

Let $\aars = \langle A, {\to} \rangle$ be an ARS.

\begin{definition}[Rewrite sequence]
  A \emph{rewrite sequence} (or \emph{reduction}) in $\aars$ is a sequence of rewrite steps $a_0 \to a_1 \to a_2 \to \cdots$
  where $a_0, a_1,\ldots \in A$.
  If the rewrite sequence is finite,
  say $a_0 \to a_1 \to \cdots \to a_n$ for $n \ge 0$,
  the sequence is said to have length $n$ and to end in $a_n$.
  If $a_0 = a_n$ and $n > 0$, then we speak of a \emph{reduction cycle}. 
  If there is a finite rewrite sequence from $a$ and to (ending in) $b$ we write $a \to^* b$ or $a \mred b$.
\end{definition}

\begin{definition}
  An object $a \in A$ is a \emph{normal form} if there is no $b \in A$ such that $a \to b$.
\end{definition}

\begin{definition}
  The ARS $\aars$ is said to be
  \begin{enumerate}
    \item
      \emph{terminating} (or \emph{strongly normalizing}), denoted $\SNr{\aars}$, if there are no infinite reductions;
    \item \emph{weakly normalizing} if for every $a \in A$ there exists a normal form $b$ such that $a \mred b$; 
    \item
      \emph{confluent} (has the \emph{Church-Rosser property}), denoted $\CRr{\aars}$, if ${\mredi \cdot \mred} \subseteq {\mred \cdot \mredi}$;
    \item
      \emph{acyclic} if there are no reduction cycles.
  \end{enumerate}
\end{definition}

\subsection{Term rewriting}
We introduce some basic notions of term rewriting and unranked terms.
For background on term rewriting with unranked terms, we refer to~\cite{kutsia2002, kutsia2004}.
For general background in (term) rewriting we refer to Terese~\cite{terese}, Klop~\cite{klop1993term}, Dershowitz-Jouannaud~\cite{dershowitz1991rewrite}, Baader-Nipkow~\cite{baad:nipk:1998}.

Let $\Sigma$ be a (possibly infinite) set of symbols, called the \emph{(unranked) signature}. 
Let $\avars = \{x,y,z,\ldots \}$ be an infinite set of variables such that $\avars \cap \Sigma = \emptyset$. 

\begin{definition}[Terms]
  The set of \emph{(unranked) terms $\ter{\Sigma}{\avars}$} over the signature $\Sigma$ and set of variables $\avars$ is inductively defined by:
  \begin{enumerate}[label=(\roman*)]
    \item 
      $\avars \subseteq \ter{\Sigma}{\avars}$, and
    \item 
      $f(t_1,\ldots,t_n) \in \ter{\Sigma}{\avars}$ whenever $f \in \Sigma$ and $t_1,\ldots,t_n \in \ter{\Sigma}{\avars}$.
  \end{enumerate}
  We denote terms $f()$ by $f$ for short. 
\end{definition}
\noindent
As we only consider unranked terms,
we will simply speak of terms. 

\begin{remark}\label{ex:tree:ordered}
  The set of terms $\ter{\Sigma}{\avars}$ is isomorphic to the set of ordered, rooted finite trees with inner nodes labeled by elements from $\Sigma$ and leaves labelled by elements from $\Sigma \cup \avars$. Here ``ordered'' means that there is an order between node siblings. In illustrations, the sibling-order is the left-to-right. 
  We frequently depict terms by ordered trees.
  For instance,
  
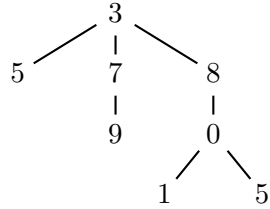
\begin{center}
    \begin{tikzpicture}[arsdefault,level distance=8mm,nodes={circle,outer sep=0.5mm},sibling distance=13mm]
  \node (3) {$3$}
      child { node (5) {$5$}
    }
      child { node (7) {$7$}
        child { node (9) {$9$}
      }
    }
      child { node (8) {$8$}
        child { node (0) {$0$}
          child { node (1) {$1$}
        }
          child { node (5) {$5$}
        }
      }
    };
  \begin{pgfonlayer}{background}
  \end{pgfonlayer}
  
    \end{tikzpicture}
  \end{center}  
  depicts the term $3(5,7(9),8(0(1,5)))$.
\end{remark}

\begin{definition}[Variables]
  The set of \emph{variables $\vars{t} \subseteq \avars$ of $t \in \ter{\Sigma}{\avars}$} is given by:
  \begin{align*}
    \vars{x} &= \{x\} &
    \vars{f(t_1,\ldots,t_n)} &= \vars{t_1} \cup \cdots \cup \vars{t_n}
  \end{align*}
\end{definition}

\begin{definition}[Positions]
  The set of \emph{positions $\pos{t}\subseteq \nat^*$ of $t \in \ter{\Sigma}{\avars}$} is defined by:
  \begin{align*}
    \pos{x} &= \{\posemp\} &
    \pos{f(t_1,\ldots,t_n)} &= \{\posemp\} \cup \{ip \mid 1\le i\le n,\;p \in \pos{t_i}\}
  \end{align*}
  For $p \in \pos{t}$, the \emph{subterm $\subtermat{t}{p}$ of $t$ at position $p$} is defined by:
  \begin{align*}
    \subtermat{t}{\posemp} &= t &
    \subtermat{f(t_1,\ldots,t_n)}{ip} &= \subtermat{t_i}{p}
  \end{align*}
\end{definition}

\begin{definition}[Substitutions]
  A \emph{substitution $\sigma$} is a map $\sigma : \avars \to \ter{\Sigma}{\avars}$.
  We extend the domain of substitutions $\sigma$ to terms $\ter{\Sigma}{\avars}$ as follows:
  \begin{align*}
    \sigma(f(t_1,\ldots,t_n)) &= f(\sigma(t_1),\ldots,\sigma(t_n))
  \end{align*}
  We write $\{\,x_1 \mapsto t_1,\ldots,\,x_n\mapsto t_n\,\}$ for the substitution $\sigma$ given by
  $\sigma(x_1) = t_1$, \ldots, $\sigma(x_n) = t_n$ and $\sigma(y) = y$ for every $y \in \avars\setminus\{\,x_1,\ldots,x_n\,\}$.
\end{definition}

\begin{definition}[Contexts]
  A \emph{context $C$} is a term $\ter{\Sigma}{\avars\cup\{\cxthole\}}$ containing precisely one occurrence of~$\cxthole$, that is, there is precisely one position $p \in \pos{C}$ such that $\subtermat{C}{p} = \cxthole$. For a context $C$ and a term $t$, we write $C[t]$ for the term $\{\,\cxthole \mapsto t\,\}(C)$.
\end{definition}

\begin{definition}[Rewriting]
  A \emph{rewrite rule $\ell \to r$} over $\Sigma$ and $\avars$ is a pair $(\ell,r) \in \ter{\Sigma}{\avars} \times \ter{\Sigma}{\avars}$ of terms such that the left-hand side $\ell$ is not a variable ($\ell \not\in \avars$), and all variables in the right-hand side $r$ occur in $\ell$ ($\vars{r} \subseteq \vars{\ell}$). 
  
  A \emph{term rewrite system (TRS) $\atrs$} over $\Sigma$ and $\avars$ is a set of rewrite rules over $\Sigma$ and $\avars$.
  The \emph{rewrite relation} ${\to_{\atrs}} \subseteq {\ter{\Sigma}{\avars} \times \ter{\Sigma}{\avars}}$ \emph{generated by $\atrs$} consist of all pairs
  \begin{align*}
    C[\sigma(\ell)] \to_{\atrs} C[\sigma(r)]
  \end{align*} 
  for contexts $C$, rules $\ell\to r \in \atrs$, and substitutions $\sigma : \avars \to \ter{\Sigma}{\avars}$. Furthermore, we write~$\redrat{\atrs}{p}$ whenever additionally $\subtermat{C}{p} = \cxthole$.
  We drop the subscript $\atrs$ in $\to_{\atrs}$ and $\redrat{\atrs}{p}$ whenever $\atrs$ is clear from the context.
\end{definition}

\begin{notation}[Argument vector notation]\label{notation:argument:vector}
We use expressions of the form $\vec{x}$ to denote a finite (possibly empty) sequence of distinct variables $(x_1, \ldots, x_n)$ for some $n \in \mathbb{N}$. For example, the rule schema
\[
  f(\vec{x}, y, \vec{z}) \to f(\vec{x}, \vec{z})
\]
enumerates rules
\[
f(x_1, \ldots, x_n, y, z_1, \ldots, z_m)
\to
f(x_1, \ldots, x_n, z_1, \ldots, z_m)
\]
for every $m,n \geq 0$. The rule schema thus allows deleting an arbitrary subtree $y$ under head symbol $f$.
\end{notation}

\subsection{Tree rewriting}

We introduce (finite, rooted, labelled, unordered) trees as equivalence classes of terms modulo the permutation of arguments. 
In our setting this view is convenient as we have already introduced terms, we use terms to denote trees, and we will switch from trees to terms by ``freezing'' in order to make use of positions and tracing.

\begin{definition}[Swapping]
  Let $\atrs_\commute$ be the TRS consisting of the rules
  \begin{align*}
     f(\vec{u},x,y,\vec{w}) \to f(\vec{u},y,x,\vec{w})
  \end{align*}
  for every $f \in \Sigma$. 
  We write $\commutestep$ for the rewrite relation $\to_{\atrs_\commute}$,
  and we use $\commute$ to denote the equivalence relation generated by $\commutestep$.
\end{definition}
Note that $\commute$ coincides with the transitive closure of $\commutestep$, that is, ${\commute} = {\commutestep^*}$.

\begin{definition}[Trees]
  The set of (unordered) \emph{trees} $\tre{\Sigma}{\avars}$ over $\Sigma$ and $\avars$ is $\quotient{\ter{\Sigma}{\avars}}{\commute}$, 
  the equivalence classes of terms modulo $\commute$.
  For a term $t \in \ter{\Sigma}{\avars}$, we write $\treecommute{t}$ 
  for the equivalence class of $t$ with respect to $\commute$.
  So $\tre{\Sigma}{\avars} = \set{\treecommute{t} \mid t \in \ter{\Sigma}{\avars}}$.
\end{definition}

\begin{example}
  In a tree the subtrees of a node may be permuted without changing the tree.
  So $\treecommute{3(5,7(9),8(0(1,5)))}$ denotes the same tree as  $\treecommute{3(8(0(5,1)),5,7(9))}$.
\end{example}

In this paper we speak of terms whenever the arguments are ordered, and we speak of trees when the arguments  are unordered.
Most of this paper will be concerned with trees.
In the proof in Section~\ref{sec:star:termination}, however, we ``freeze'' trees, obtaining terms, in order to make use of positions and tracing throughout rewrite sequences. This freezing procedure was also used in~\cite{klop1993term}.

\begin{remark}
  We depict trees in the same way as we depict terms:
  
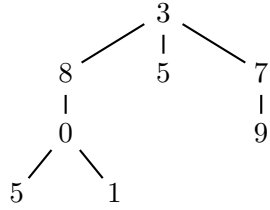
\begin{center}
    \begin{tikzpicture}[arsdefault,level distance=8mm,nodes={circle,outer sep=0.5mm},sibling distance=13mm]
  \node[xshift=0mm] (3) {$3$}
      child { node (8) {$8$}
        child { node (0) {$0$}
          child { node (5) {$5$}
        }
          child { node (1) {$1$}
        }
      }
    }
      child { node (5) {$5$}
    }
      child { node (7) {$7$}
        child { node (9) {$9$}
      }
    };
  \begin{pgfonlayer}{background}
  \end{pgfonlayer}
  
    \end{tikzpicture}
  \end{center}
  It will be clear from the context whether such illustration 
  depicts a term (ordered) or a tree (unordered).
\end{remark}

We rewrite trees using term rewriting systems.

\begin{definition}
  Let $\atrs$ be a term rewriting system.
  We extend $\to_\atrs$ from terms $\ter{\Sigma}{\avars}$ to trees $\tre{\Sigma}{\avars}$ by defining
  $\treecommute{s} \to_\atrs \treecommute{t}$
  whenever $s,t \in \ter{\Sigma}{\avars}$ and $s \to_\atrs t$.
\end{definition}
\noindent
So $\set{(\treecommute{s}, \treecommute{t}) \mid s,t \in \ter{\Sigma}{\avars},\; s \to_\atrs t}$ is the rewrite relation induced by $\atrs$ on trees. 

A term rewrite system $\atrs$ thus induces abstract rewrite systems on terms and trees, namely $\langle \ter{\Sigma}{\avars},\, \to_{\atrs} \rangle$
and $\langle \tre{\Sigma}{\avars},\, \to_{\atrs} \rangle$, respectively.
Hence, the notions of rewrite sequences, reduction cycles, normal forms, termination, confluence, acyclicity and so forth, carry over from abstract rewriting to term and tree rewriting.

\subsection{Orders}

We also presuppose some familiarity with the basics of order theory, partial orders, well-founded orders, and trees. Particularly important in this paper is the notion of a \emph{well-quasi-order} (\emph{wqo}), an order that often proves fruitful in the termination setting ~\cite{weiermann2020wqo}.

\begin{definition}(Well-quasi-order)
  A \emph{quasi-order} is a binary relation that is reflexive and transitive. A \emph{well-quasi-order (wqo)} is a quasi-order $\trianglelefteq$ such that any infinite sequence of elements $a_0, a_1, \ldots, a_i, \ldots, a_j, \ldots$ contains an increasing pair $a_i \trianglelefteq a_j$ with $i < j$.
\end{definition}

\begin{proposition}
Any wqo is well-founded.
\end{proposition}

\subsection{Ordinals}

An \emph{ordinal}~\cite{rubin1967set} is a transitive set whose elements are again transitive sets. 
A set~$A$ is \emph{transitive} if $x \subseteq A$ for every $x \in A$.
In this view, every ordinal $\alpha$ is the set of all ordinals smaller than $\alpha$; thus $\alpha = \set{ \beta \mid \beta < \alpha }$.
So the class of all ordinals is totally ordered by the membership relation.
The least ordinal is $\set{}$, called $0$.
The \emph{successor} $\alpha^+$ of an ordinal $\alpha$ is the ordinal $\alpha^+ = \alpha \cup \set{ \alpha }$.
A \emph{limit} ordinal is an ordinal that is neither $0$ nor a successor ordinal. 
So the smallest ordinals are
\begin{align*}
    0 = \set{},\; 
    1 = 0^+ = \set{\set{}},\;
    2 = 1^+ = \set{\set{}, \set{\set{}}},\quad
    \ldots,\quad
    \omega,
    \ldots
\end{align*}
The ordinal $\omega$ is the least infinite ordinal, $\omega = \set{0,1,2,3,\ldots}$.

The Kirby-Paris Hydra and the Buchholz Hydra treated in this paper have traditionally been studied and used in the context of ordinal analysis~\cite{rathjen2006art}, a branch of proof theory that measures the strength of mathematical theories in terms of ordinals. More specifically, the \emph{ordinal strength} (or \emph{proof-theoretic ordinal}) of such a theory $T$ can be characterized as the smallest ordinal that $T$ cannot prove well-founded. In this paper, we will say that a Hydra has \emph{ordinal strength} $\alpha$ if a proof of its termination requires a theory with ordinal strength at least $\alpha$.

While we will not perform any ordinal analysis in this paper, we will ask questions pertaining to it, in particular in relation to our newly introduced Star Hydra. For this reason, we provide the following overview of large ordinals, assuming only an understanding of the limit ordinal $\omega$ and ordinal arithmetic~\cite{rubin1967set}.

\newcommand{\rddots}{\cdot^{\cdot^{\cdot}}}
\newcommand{\lddots}{\cdot_{\cdot_{\cdot}}}

\subsubsection{Ordinals up to $\epsilon_0$}\label{sec:ordinals:up:to:epsilon}

Any ordinal can be written as a polynomial 
$$\omega^{\alpha_1}c_1 + \omega^{\alpha_2}c_2 + \cdots + \omega^{\alpha_n}c_n,$$
where $\alpha_1 > \alpha_2 > \cdots > \alpha_n \geq 0$
are ordinals and the $c_i$ are positive integers. This representation is called Cantor normal form (CNF).

Using ordinals $\leq \omega$ as elementary symbols, one can express ordinals lying on the spectrum $\omega, \omega^\omega, \omega^{\omega^\omega}, \ldots$ in CNF. The limit ordinal for this sequence is the first solution to the fixed point equation $\alpha = \omega^\alpha$, called $\epsilon_0$. The weakest Hydra treated in this paper, the Kirby-Paris Hydra, has ordinal strength $\epsilon_0$.

The equation $\alpha = \omega^\alpha$ in fact has infinitely many fixed points. The $\iota$-th ordinal that is a solution to $\alpha = \omega^\alpha$ is called $\epsilon_\iota$. If we similarly generalize the subscript of $\epsilon$, the limit for what we can express is the first fixed point solution to $\alpha = \epsilon_\alpha$. 
Call this ordinal $\zeta_0$. 
Once again generalizing the subscript, we can express ordinals up to the first fixed point solution to $\alpha = \zeta_\alpha$. 
If we were to continue in this fashion, we would soon exhaust the Greek alphabet.
Veblen functions provide a systematic solution to this problem.

\subsubsection{Ordinals up to $\Gamma_0$}
The Veblen hierarchy $\phi_\alpha(\gamma)$ can be defined by two clauses:
\begin{itemize}
\item for $\alpha = 0$, $\phi_\alpha(\gamma) = \omega^\gamma$, and
\item for $\alpha > 0$,
$\phi_\alpha(\gamma)$ is the $(1+\gamma)$-th common fixed point of the functions $\phi_\beta$ for all $\beta < \alpha$.
\end{itemize}

So in particular, we obtain $\phi_1(0) = \epsilon_0$ and $\phi_2(0) = \zeta_0$.
The first solution to the fixed point equation $\phi_\alpha(0) = \alpha$ is the Feferman–Sch\"utte ordinal, and it is denoted by $\Gamma_0$ \cite{veblen1908, schutte1965}. 

\subsubsection{Ordinals beyond $\Gamma_0$}

Multiple definitions exist that allow us to go beyond the Veblen hierarchy.
Buchholz's $\psi$ functions~\cite{buchholz1986psifunctions}  are one example.
The ordinal strength of the strongest Hydra in our paper, the powerful Buchholz Hydra, exceeds even all the ordinals expressible using the $\psi$ functions.

\section{Star Games}\label{sec:star:games}

Proving termination (\SN{}) of various rewriting sequences is of major importance. Not only is it a central aspect of program correctness, it is also a stepping stone towards confluence, using Newman's lemma.
Although it is in general undecidable whether a (string or term) rewrite system is terminating, in many instances termination can be proven, 
and various techniques have been developed to do so.

In this paper we present one of the most powerful of such termination proof techniques:
the \emph{recursive path order}~\cite{dershowitz1982orderings} as developed by Dershowitz.
We will use the presentation as in~\cite{klop2006iterative},
where the inductive definitions of the usual presentation
are replaced by a reduction procedure that is easier to grasp to our taste.
We develop this application, employing trees with annotations, 
natural numbers or ordinals at the nodes, and an auxiliary marker $\star$ at some nodes.
It turns out that we will arrive at vast generalizations of the well-known \emph{multiset orders}~\cite{manna1980lectures, oostrom1994phd}.

\begin{figure}
  \centering
    \begin{tikzpicture}[arsdefault,level distance=8mm,nodes={circle,outer sep=0.5mm},sibling distance=13mm]
\node (3^star) {$3^\star$}
    child { node (5) {$5$}
  }
    child { node (7) {$7$}
      child { node (9^star) {$9^\star$}
    }
  }
    child { node (8^star) {$8^\star$}
      child { node (0) {$0$}
        child { node (1) {$1$}
      }
        child { node (5) {$5$}
      }
    }
  };
\begin{pgfonlayer}{background}
\end{pgfonlayer}

    \end{tikzpicture}
  \caption{A tree with stars.}
  \label{fig:tree:stars}
\end{figure}
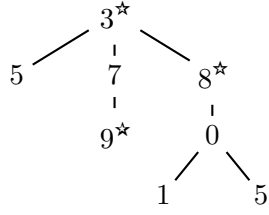

\begin{definition}[Starred signature]
$\Sigma^\star = \{ f^\star \mid f \in \Sigma \}$ is the \emph{starred signature}. It is assumed that $\Sigma$ is unstarred, i.e., it does not contain a symbol of the form $x^\star$.
\end{definition}

\newcommand{\terstar}[2]{\ter{#1}{#2}^\star}
\newcommand{\trestar}[2]{\tre{#1}{#2}^\star}
\begin{notation}
We abbreviate $\ter{\Sigma \cup \Sigma^\star}{\avars}$ as $\terstar{\Sigma}{\avars}$, and
$\tre{\Sigma \cup \Sigma^\star}{\avars}$
as $\trestar{\Sigma}{\avars}$.
\end{notation}

\begin{definition}[Iterative path order]\label{def:star:game}%
  Let ${>} \subseteq \Sigma \times \Sigma$ be a well-founded partial order, often called the \emph{precedence ordering}.
  The \emph{iterative path order (IPO)} on trees $\tre{\Sigma}{\avars}$ is 
  $${\ilpored^+} \cap {(\tre{\Sigma}{\avars} \times \tre{\Sigma}{\avars})}\;.$$
  Here $\ilpored$ is the rewrite relation induced on $\tre{\Sigma}{\avars}$ by the TRS $\trsstar$ consisting of the rules
  \begin{gather*}
    \begin{aligned}
      (\lput) && f(\vec{x}) &\ilpored f^\star(\vec{x}) \\
      (\lselect) && f^\star(\vec{x},y,\vec{z}) &\ilpored y  && \\
      (\lcopy) && f^\star(\vec{x}) &\ilpored g(\underbrace{f^\star(\vec{x}),\ldots,f^\star(\vec{x})}_{\text{$0$ or more}}) && \text{if $f > g$} \\
      (\ldown) && f^\star( \vec{x},g(\vec{y}),\vec{z} ) &\ilpored f( \vec{x}, \underbrace{g^\star(\vec{y}), \ldots, g^\star(\vec{y})}_{\text{$0$ or more}}, \vec{z} )
    \end{aligned}
  \end{gather*}
  for every $f,g \in \Sigma$.
  
  Overloading names and notation, we also use $\ilpored$ to denote the rewrite relation on  $\terstar{\Sigma}{\avars}$ induced by these rules.
\end{definition}

\begin{remark}
Although they are rule schemas in the technical sense, we will often refer to $\lput$, $\lselect$, $\lcopy$, and $\ldown$ simply as rules. Similarly for many other rule schemas in this paper.
\end{remark}

The idea behind $\trsstar$ is that marking the root symbol of a tree, 
by means of the $\lput$ rule, corresponds to the obligation to make that tree \emph{smaller},
whereas the other rules correspond to ways in which this can be brought about:
\begin{enumerate}
  \item
    The $\lselect$ rule expresses that selecting one of the children of a tree makes it smaller.
  \item
    The $\lcopy$ rule expresses that a tree $t$ can be made smaller
    by putting finitely many copies of trees smaller than $t$ 
    below a root symbol $g$ which is less heavy than the root symbol $f$ of $t$.
  \item
    The $\ldown$ rule expresses that a tree can be made smaller by replacing one of its subtrees with finitely many smaller subtrees.
\end{enumerate}

\begin{remark}[$\ldown$ vs.\ $\llex$]\label{rem:lex}
  In the literature on iterative lexicographic path orders, one frequently finds the $\llex$ rule~\cite{klop1993term,gese:1990,klop2006iterative,kop2009iterative} instead of the $\ldown$ rule:
  \begin{align*}
      (\llex) && f^\star( \vec{x},g(\vec{y}),\vec{z} ) &\ilpored f( \vec{x}, g^\star(\vec{y}), \underbrace{\ell, \ldots, \ell}_{\text{$0$ or more}} )
      && \text{where $\ell = f^\star( \vec{x},g(\vec{y}),\vec{z} )$}
  \end{align*}
  The $\llex$ rule is stronger than the $\ldown$ rule in the sense that the latter can be derived using the former (in combination with the rules $\lput$ and $\lselect$).
  However, the $\llex$ rule is only valid on terms,
  but it leads to non-termination in the setting of (unordered) trees.
  The $\ldown$ rule, on the other hand, is valid also on trees, and this is the setting that we are focusing on in this paper.
\end{remark}

\begin{example}[A cyclic reduction]\label{example:cyclic}
  A cyclic reduction
  \medskip
  \begin{center}
    {
\newcommand{\F}{$\fun{f}$}
\newcommand{\A}{$\fun{a}$}
\newcommand{\B}{$\fun{b}$}
\newcommand{\G}{$\fun{g}$}
\renewcommand{\H}{$\fun{h}$}
\begin{tikzpicture}[default,thick,level distance=8mm,sibling distance=10mm,inner sep=1mm,
  every node/.style={rectangle}]

\node[xshift=0mm] (t1) {\F}
    child { node (A) {\A}
  }
    child { node (G) {\G}
      child { node (B) {\B}
    }
  };
\begin{pgfonlayer}{background}
\end{pgfonlayer}

  \node [startAt=t1] {{\large $\star$}};

  \begin{scope}[medium tree]
\node[xshift=33mm] (t2) {\H}
    child { node (Fa) {\F}
      child { node (A) {\A}
    }
      child { node (G) {\G}
        child { node (B) {\B}
      }
    }
  }
    child { node (Fb) {\F}
      child { node (A) {\A}
    }
      child { node (G) {\G}
        child { node (B) {\B}
      }
    }
  };
\begin{pgfonlayer}{background}
\end{pgfonlayer}

  \end{scope}

  \node [startAt=Fa] {{\large $\star$}};
  \node [startAt=Fb] {{\large $\star$}};
  
  \begin{scope}[medium tree]
\node[xshift=74mm] (t3) {\H}
    child { node (Fa) {\F}
      child { node (A) {\A}
    }
      child { node (G) {\G}
        child { node (B) {\B}
      }
    }
  }
    child { node (Fb) {\F}
      child { node (A) {\A}
    }
      child { node (G) {\G}
        child { node (B) {\B}
      }
    }
  };
\begin{pgfonlayer}{background}
\end{pgfonlayer}

  \end{scope}

  \node [startAt=t3] {{\large $\star$}};
  \node [startAt=Fa] {{\large $\star$}};
  \node [startAt=Fb] {{\large $\star$}};

\node[xshift=110mm] (t4) {\F}
    child { node (A) {\A}
  }
    child { node (G) {\G}
      child { node (B) {\B}
    }
  };
\begin{pgfonlayer}{background}
\end{pgfonlayer}

  \node [startAt=t4] {{\large $\star$}};

  \begin{scope}[very thick,->,n/.style={at end,anchor=north west,xshift=-2mm,inner sep=.75mm}]
    \draw ($(t1)!.3!(t2)$) -- node [n] {$\lcopy$} ++(7mm,0mm);
    \draw ($(t2)!.3!(t3)$) -- node [n] {$\lput$} ++(7mm,0mm);
    \draw ($(t3)!.3!(t4)$) -- node [n] {$\lselect$} ++(7mm,0mm);
  \end{scope}
\end{tikzpicture}
}
    \smallskip
  \end{center}
  in $\trsstar$. (Here $ h < f$.)
\end{example}

Clearly, $\trsstar$ is neither \SN, nor \CR.
However, we are interested in termination on star-free trees.
It turns out that $\trsstar$ does not admit infinite rewrite sequences
that contain an infinite number of star-free trees.
In other words, the restriction of the transitive closure $\ilpored^+$ to $\tre{\Sigma}{\avars}$  is terminating.

\begin{theorem}\label{thm:rpo:tree:sn}
  The relation $\ilpored^+$, restricted to $\tre{\Sigma}{\avars}$, is $\SN$.
\end{theorem}

\noindent
We will present two proofs of this theorem.

The first proof makes use of Kruskal's tree theorem in combination with acyclicity of the reduction relation $\ilpored$.
This route has also been used in the seminal paper~\cite{dershowitz1982orderings} introducing the recursive path order. Nevertheless, the proof that we present has a very different flavour. In~\cite{dershowitz1982orderings} the recursive path ordering
is defined via an inductive definition, and the crucial step in the termination proof is establishing transitivity of the defined order. In our setting, transitivity of the iterative path order is immediate since the order is defined via the transitive closure of $\ilpored$. Our proof, presented in Section~\ref{sec:scenic:route}, establishes acyclicity by tracing symbol occurrences throughout $\ilpored$ reductions.

The second proof, presented in Section~\ref{sec:buchholz}, is based on well-founded induction, as in~\cite{buchholz1995proof,klop2006iterative}. This proof gives a stronger result than the first proof in the presence of infinite signatures $\Sigma$, as it establishes termination of the iterative path order for arbitrary well-founded partial orders $(\Sigma,<)$. For the first proof, $(\Sigma,\le)$ needs to be a wqo due to the use of Kruskal's tree theorem. 

Theorem~\ref{thm:rpo:tree:sn} is of significant practical interest because of the following corollary, which intuitively states that a relation $R$ is terminating if it can be simulated by $\ilpored^+$.

\begin{corollary}[Termination by simulation]\label{cor:termination:by:simulation}
  If $R \subseteq {\tre{\Sigma}{\avars} \times \tre{\Sigma}{\avars}}$
  and $R \subseteq {\ilpored^+}$,
  then $\SNr{R}$.
\end{corollary}

We illustrate this corollary by applying it to a simple word rewrite system.

\begin{example}[$\trsstar$ in practice]\label{ex:irpo_application}
  Consider the rewrite system with signature $\Sigma = \{1, 0\}$ and the single rewrite rule 
  \begin{align*}
     1(0(x)) \to 0(0(1(x)))\;. 
  \end{align*}
  By defining the precedence ordering $1 > 0$, we obtain the IPO reduction
  \[
  \begin{array}{llll}
    \treecommute{1(0(x))} & \ilpored_\lput & \treecommute{1^\star(0(x))} \\
    & \ilpored_\lcopy & \treecommute{0(1^\star(0(x)))} \\
    & \ilpored_\lcopy & \treecommute{0(0(1^\star(0(x))))} \\
    &\ilpored_\ldown & \treecommute{0(0(1(0^\star(x))))} \\
    &\ilpored_\lselect & \treecommute{0(0(1(x)))}
  \end{array}
  \]
  on unstarred trees. Since the conditions of Corollary~\ref{cor:termination:by:simulation} are met, it follows that $\to$ is terminating.
\end{example}

A more interesting example is given in Section~\ref{sec:hydras}, where $\trsstar$ is used to show that the Kirby-Paris Hydra is terminating.

\section{The First Proof: The Scenic Route}\label{sec:scenic:route}

In this section, we present the first proof of Theorem~\ref{thm:rpo:tree:sn} based on Kruskal's tree theorem, which presents a well-known way to establishing termination of the star games~\cite{klop1993term}. Due to the use of Kruskal's tree theorem, this proof requires that $(\Sigma,\le)$ is a well-quasi-order (wqo).
The proof in this section is centered around the following simple, but pivotal, proposition.

\begin{proposition}\label{basic}
  Let $\mathcal{A} = \langle A, \to \rangle$ be an ARS and suppose that
  \begin{enumerate}[label=(\roman*)]
    \item\label{basic:wqo} ${\trianglelefteq} \subseteq A \times A$ is a wqo,
    \item $ a \trianglelefteq b \implies b \twoheadrightarrow a$, and
    \item\label{basic:ac} $ \to $ is acyclic.
  \end{enumerate}
  Then the reduction $\to$ is SN.
\end{proposition}

\begin{proof}
  Assume for a contradiction that there is an infinite rewrite sequence $a_0 \to a_1 \to \cdots$. Then since $\trianglelefteq$ is a wqo on $A$, there exists a pair $a_i \to^+ a_j$ with $a_i \trianglelefteq a_j$. Thus by (2), $a_j \twoheadrightarrow a_i$. Hence $a_i \to^+ a_i$, contradicting that $\to$ is acyclic. So $\to$ must be SN.
\end{proof}

\noindent
The idea behind Proposition~\ref{basic} is illustrated in Figure~\ref{tree-barrier}.\footnote{
  Proposition~\ref{basic} can be viewed as a generalization of \cite[Theorem 4.4]{midd:zant:1997}, the termination of simplification orders, to abstract reduction systems. In~\cite[Theorem 4.4]{midd:zant:1997}, the objects are terms, and the wqo $\trianglelefteq$ is the homeomorphic embedding.
}

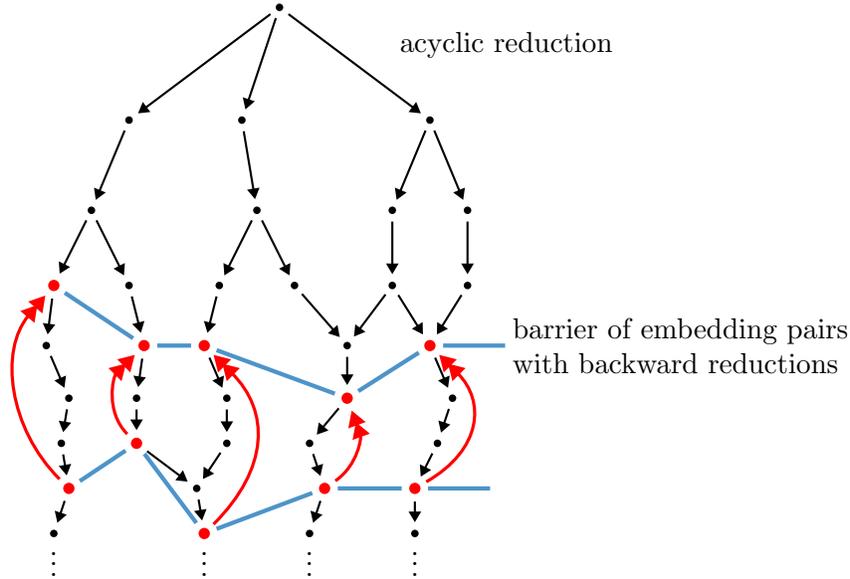
\begin{figure}
  \centering
  \begin{tikzpicture}[default,nodes={smallCircle},->,r/.style={fill=cred,minimum size=1.5mm}]
    \node (r) {};
    \node (r1) at ($(r) + (-20mm,-15mm)$) {}; \draw (r) -- (r1);
      \node (r11) at ($(r1) + (-5mm,-12mm)$) {}; \draw (r1) -- (r11);
        \node [r] (r111) at ($(r11) + (-5mm,-10mm)$) {}; \draw (r11) -- (r111);
          \node (r1111) at ($(r111) + (-1mm,-8mm)$) {}; \draw (r111) -- (r1111);
            \node (r11111) at ($(r1111) + (3mm,-7mm)$) {}; \draw (r1111) -- (r11111);
              \node (r111111) at ($(r11111) + (-1mm,-6mm)$) {}; \draw (r11111) -- (r111111);
                \node [r] (r1111111) at ($(r111111) + (1mm,-6mm)$) {}; \draw (r111111) -- (r1111111);
                  \node (r11111111) at ($(r1111111) + (-2mm,-6mm)$) {}; \draw (r1111111) -- (r11111111);
                  \node [draw=none,fill=none] at ($(r11111111) + (0mm,-3mm)$) {$\vdots$};
        \node (r112) at ($(r11) + (5mm,-10mm)$) {}; \draw (r11) -- (r112);
          \node [r] (r1121) at ($(r112) + (2mm,-8mm)$) {}; \draw (r112) -- (r1121);
            \node (r11211) at ($(r1121) + (-1mm,-7mm)$) {}; \draw (r1121) -- (r11211);
              \node [r] (r112111) at ($(r11211) + (0mm,-6mm)$) {}; \draw (r11211) -- (r112111);
                \node (r1121111) at ($(r112111) + (8mm,-6mm)$) {}; \draw (r112111) -- (r1121111);
                  \node [r] (r11211111) at ($(r1121111) + (1mm,-6mm)$) {}; \draw (r1121111) -- (r11211111);
                  \node [draw=none,fill=none] at ($(r11211111) + (0mm,-3mm)$) {$\vdots$};
    \node (r2) at ($(r) + (-5mm,-15mm)$) {}; \draw (r) -- (r2);
      \node (r21) at ($(r2) + (2mm,-12mm)$) {}; \draw (r2) -- (r21);
        \node (r211) at ($(r21) + (-5mm,-10mm)$) {}; \draw (r21) -- (r211);
          \node [r] (r2111) at ($(r211) + (-2mm,-8mm)$) {}; \draw (r211) -- (r2111);
            \node (r21111) at ($(r2111) + (3mm,-7mm)$) {}; \draw (r2111) -- (r21111);
              \node (r211111) at ($(r21111) + (0mm,-6mm)$) {}; \draw (r21111) -- (r211111);
                \draw (r211111) -- (r1121111);
              \node (r211112) at ($(r21111) + (11mm,-6mm)$) {}; %
                \node [r] (r2111121) at ($(r211112) + (2mm,-6mm)$) {}; \draw (r211112) -- (r2111121);
                  \node (r21111211) at ($(r2111121) + (-2mm,-6mm)$) {}; \draw (r2111121) -- (r21111211);
                  \node [draw=none,fill=none] at ($(r21111211) + (0mm,-3mm)$) {$\vdots$};
        \node (r212) at ($(r21) + (5mm,-10mm)$) {}; \draw (r21) -- (r212);
          \node (r2121) at ($(r212) + (7mm,-8mm)$) {}; \draw (r212) -- (r2121);
            \node [r] (r21211) at ($(r2121) + (0mm,-7mm)$) {}; \draw (r2121) -- (r21211);
              \draw (r21211) -- (r211112);
    \node (r3) at ($(r) + (20mm,-15mm)$) {}; \draw (r) -- (r3);
      \node (r31) at ($(r3) + (-5mm,-12mm)$) {}; \draw (r3) -- (r31);
        \node (r311) at ($(r31) + (0mm,-10mm)$) {}; \draw (r31) -- (r311);
          \draw (r311) -- (r2121);
      \node (r32) at ($(r3) + (5mm,-12mm)$) {}; \draw (r3) -- (r32);
        \node (r321) at ($(r32) + (0mm,-10mm)$) {}; \draw (r32) -- (r321);
          \node [r] (r3211) at ($(r321) + (-5mm,-8mm)$) {}; \draw (r321) -- (r3211);
          \draw (r311) -- (r3211);
            \node (r32111) at ($(r3211) + (3mm,-7mm)$) {}; \draw (r3211) -- (r32111);
              \node (r321111) at ($(r32111) + (-2mm,-6mm)$) {}; \draw (r32111) -- (r321111);
                \node [r] (r3211111) at ($(r321111) + (-3mm,-6mm)$) {}; \draw (r321111) -- (r3211111);
                  \node (r32111111) at ($(r3211111) + (0mm,-6mm)$) {}; \draw (r3211111) -- (r32111111);
                  \node [draw=none,fill=none] at ($(r32111111) + (0mm,-3mm)$) {$\vdots$};
    \draw [ultra thick,cblue!70,-] (r111) -- (r1121) -- (r2111) -- (r21211) -- (r3211) -- node [rectangle,fill=none,at end,anchor=west,align=left,black] {barrier of embedding pairs\\ with backward reductions} ++(10mm,0mm);
    \draw [ultra thick,cblue!70,-] (r1111111) -- (r112111) -- (r11211111) -- (r2111121) -- (r3211111) -- ++(10mm,0mm);
    \begin{scope}[cred,->>,very thick]
      \draw (r1111111) to[out=135,in=-130] (r111);
      \draw (r112111) to[out=135,in=-135] (r1121);
      \draw (r11211111) to[out=45,in=-45,looseness=1.2] (r2111);
      \draw (r2111121) to[out=35,in=-70] (r21211);
      \draw (r3211111) to[out=30,in=-45,looseness=1.3] (r3211);
    \end{scope}
    \node [anchor=west,rectangle,fill=none] at (15mm,-5mm) {acyclic reduction};
  \end{tikzpicture}
  \caption{\textit{A wqo provides a barrier for an acyclic reduction, if it subsumes the converse of the wqo embedding $\trianglelefteq$. The well-founded part of the reduction tree cut off by the wqo barrier has an ordinal height that measures how soon the barrier with its embedded pairs is manifesting itself.}}
  \label{tree-barrier}
\end{figure}

The remainder of this section is devoted to establishing conditions \ref{basic:wqo}--\ref{basic:ac} of Proposition~\ref{basic} for $A = \tre{\Sigma}{\avars}$ and ${\to} = {\ilpored^+} \cap {(\tre{\Sigma}{\avars} \times\tre{\Sigma}{\avars})}$. Then termination of ${\ilpored^+}$ on unstarred trees follows. 
In Subsection~\ref{sec:kruskal}, we recall a well-known embedding relation ${\treeemb}$ on trees, which is a wqo as stated by the beautiful Kruskal's tree theorem~\cite{nash1967infinite,rathjen1993proof,weiermann1994complexity,gallier1991s}. In Subsection~\ref{sec:embeddings:and:ipo}, we clarify the relation between $\treeemb$ and $\ilpored^+$, and show that $s \treeemb t \implies t \ilpored^+ s$. In Subsection~\ref{sec:star:termination}, finally, we demonstrate that $\ilpored^+$ is acyclic on unstarred trees by means of an elaborate symbol tracing proof.

\subsection{A reminder of Kruskal's tree theorem}\label{sec:kruskal}

\newcommand{\descendant}{\ll}

A standard embedding relation on trees exists, called the \emph{homeomorphic embedding}, for which we need the following auxiliary definitions.
Let $\alpha$ and $\beta$ be nodes in a tree. 
We write $\alpha \descendant \beta$ if $\beta$ lies within the subtree rooted at $\alpha$ and $\alpha \ne \beta$.
We write $\alpha \wedge \beta$ for the supremum of $\alpha$ and $\beta$ with respect to $\descendant$.
In other words, $\alpha \wedge \beta$ is the unique node that lies 
on the intersection of the path from $\alpha$ to $\beta$, 
the path from $\alpha$ to the root and 
the path from $\beta$ to the root.

\begin{definition}[Tree embedding~\cite{gallier1991s}]\label{lem:tree:embedding}%
A tree $s \in \tre{\Sigma}{\avars}$ is \emph{homeomorphically embedded in} tree $t \in \tre{\Sigma}{\avars}$, notation $s \treeemb t$, if there is an injective function $\varphi : \treenodes{s} \to \treenodes{t}$  such that
  \begin{enumerate}
    \item \label{it:kruskal:tree:embedding:sup} $\varphi$ is supremum preserving ($\varphi(\alpha \wedge \beta) = \varphi(\alpha) \wedge \varphi(\beta)$), and
    \item $\varphi$ is label increasing ($\treelabel{\alpha} \leq \treelabel{\varphi(\alpha)}$).
  \end{enumerate}
  Here $\le$ is the wqo on the node labels $\Sigma$.
  
\end{definition}

For the labels, we will typically use natural numbers with the usual order $\le$.

\begin{example}
  We have $\treecommute{2(9,7(0,4))} \treeemb \treecommute{1(3(8(8(5,1)),9,5(9)),2)}$ due to the mapping:
  
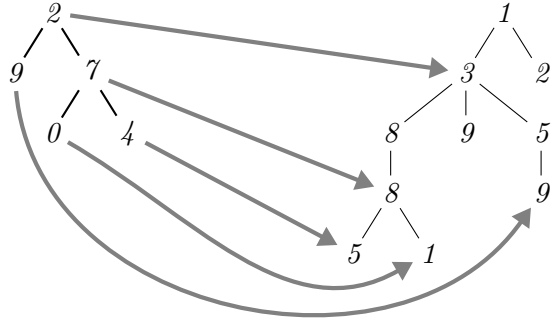
\begin{center}
    
\vspace{1ex}
\begin{tikzpicture}[default,thick,level distance=8mm,sibling distance=10mm,inner sep=.5mm,
  paint/.style={very thick,draw=#1!50!black,fill=#1,opacity=.3},
  cnode/.style={every node/.style={circle,draw=#1!50!black,fill=#1!50}}]

\node (2) {2}
    child { node (9) {9}
  }
    child { node (7) {7}
      child { node (0) {0}
    }
      child { node (4) {4}
    }
  };
\begin{pgfonlayer}{background}
\end{pgfonlayer}

\node[xshift=60mm] (1) {1}
    child { node (3r) {3}
      child { node (8) {8}
        child { node (8r) {8}
          child { node (5r) {5}
        }
          child { node (1r) {1}
        }
      }
    }
      child { node (9r') {9}
    }
      child { node (5) {5}
        child { node (9r) {9}
      }
    }
  }
    child { node (2r) {2}
  };
\begin{pgfonlayer}{background}
\end{pgfonlayer}

  \begin{scope}[->,ultra thick,gray]
  \draw (2) -- (3r);
  \draw (9) to[out=-85,in=-125,looseness=1.1] (9r);
  \draw (7) -- (8r);
  \draw (0) to[out=-30,in=-150] (1r);
  \draw (4) -- (5r);
  \end{scope}
\end{tikzpicture}
\vspace{-3ex}
  \end{center}
  
\end{example}

Clearly, $\treeemb$ is a partial order on 
$\tre{\Sigma}{\avars}$. Moreover it satisfies the following remarkable property:

\begin{theorem}[Kruskal's tree theorem~\cite{kruskal1960well}]\label{thm:kruskal}%
  Let $t_{0}, t_{1}, t_{2}, \ldots$ be an infinite sequence of trees in~$\tre{\Sigma}{\avars}$, and ${\leq}$ a wqo on the labels $\Sigma$.
  Then $t_{i} \treeemb t_{j}$ for some $i < j$.
\end{theorem}

\begin{remark}
  Kruskal's tree theorem can be phrased as follows:  if  $\pair{\Sigma}{{\leq}}$ is a wqo, then $\pair{\tre{\Sigma}{\avars}}{{\treeemb}}$ is a wqo.
\end{remark}

Kruskal's tree theorem generalizes Higman's lemma, a similar result originally stated for words over finite alphabets.

\begin{corollary}[Higman's lemma~\cite{higman1952ordering}]
  Let $(A, {\leq})$ be a (possibly infinite) wqo alphabet. Then the set of words $(A^*, {\leq_*})$ is a wqo, where $w \leq_* v$ if $w$ can be obtained from $v$ by decreasing and erasing an arbitrary number of letters in $v$.
\end{corollary}

\subsection{Embeddings and IPO}\label{sec:embeddings:and:ipo}

We now clarify the relation between the homeomorphic embedding ${\treeemb}$ and $\ilpored$.

First of all, we believe that the homeomorphic embedding relation can be more conveniently defined via term rewriting as follows.

\newcommand{\trsemb}{\mathcal{E}}
\begin{definition}[Tree embedding via term rewriting]\label{def:tree:embedding:transformation}%
  Define the TRS $\trsemb$ to consist of the following rules:
  \begin{gather*}
    \begin{aligned}
    \text{(\emph{removing a subtree})} && f(\vec{x},y,\vec{z}) &\to f(\vec{x},\vec{y}) & \\
    \text{(\emph{selecting a subtree})} && f(\vec{x},y,\vec{z}) &\to y &\\
    \text{(\emph{decreasing a label})} && f(\vec{x}) &\to g(\vec{x}) & \text{ if } f > g
    \end{aligned}
  \end{gather*}
  for $f, g \in \Sigma$.
  Here, $f > g$ means that $f \ge g$ and not $f \le g$.
  We use ${\to_\trsemb}$ to denote the rewrite relation
  induced by $\trsemb$ on trees.
  We say that $s$ can be \emph{homeomorphically embedded in $t$ (via tree rewriting)} if $t \to_\trsemb^* s$.
\end{definition}

\begin{remark}
  The term rewrite system $\trsemb$ can be used to define both the homeomorphic embedding on trees and on terms. Note that the rewrite system does not permute arguments. So it preserves the argument order (sibling order) when applied to terms. 
\end{remark}

\begin{lemma}[Equivalence of embedding relations]\label{lemma:equivalence:of:embedding:relations}
${\treeemb} = {\leftarrow_\trsemb^*}$.
\end{lemma}
\begin{proof}
  We provide a proof for $\subseteq$, which is the only direction relied on in this paper. 

  \newcommand{\subtree}[2]{#2|_{#1}}%
  \newcommand{\treeroot}[1]{\mathit{root}(#1)}%
  Assume $s \treeemb t$ is established by some witnessing homeomorphism $\phi$. For a tree $u$ and a node $n$ in $u$, we use $\subtree{n}{u}$ to denote the subtree of $u$ rooted at $n$. It suffices to prove that ($\dagger$) for node every $n$ in $s$, $\subtree{\phi(n)}{t} \to^* \subtree{n}{s}$. For then $t \to^* \subtree{\phi(\treeroot{s})}{t}$ using 0 or more selecting steps, and $\subtree{\phi(\treeroot{s})}{t} \to^* \subtree{\treeroot{s}}{s} = s$ by ($\dagger$). 
  We prove ($\dagger$) by structural induction on $s$. For the base case, $s$ is a leaf $n$. Then $\subtree{\phi(n)}{t} \to^* n = \subtree{n}{s}$ by removing all of the subtrees of $\phi(n)$, and decreasing its label.

  For the inductive step, $s = \treecommute{n(s_0, \ldots, s_k)}$ for some $k \in \mathbb{N}$. For every $0 \le i < j \le k$, the roots of $s_i$ and $s_j$ have $n$ as their supremum in $s$, and thus $\phi(\treeroot{s_i})$ and $\phi(\treeroot{s_j})$ have $\phi(n)$ as their supremum in $t$. This implies that $\phi(\treeroot{s_0})$,\ldots,$\phi(\treeroot{s_k})$ are each in distinct subtrees that are rooted at $\phi(n)$. 
  So $\subtree{\phi(n)}{t} = \treecommute{\phi(n)(t_0, \ldots, t_k, \ldots, t_{k + k'})}$ for some $k' \geq 0$ such that $\phi(\treeroot{s_i})$ is among the nodes of $t_i$ for all $i \leq k$. Now we obtain 
  \[
  \begin{array}{lcll}
    \subtree{\phi(n)}{t} & \to^* & \treecommute{\phi(n)(t_0, \ldots, t_k)} & \text{removing subtrees} \\
     & \to^* &  \treecommute{\phi(n)(\subtree{\phi(\treeroot{s_0})}{t}, \ldots, \subtree{\phi(\treeroot{s_k})}{t})}
     & \text{selecting in the $t_i$'s} \\
     & \to^* & \treecommute{\phi(n)(\subtree{\treeroot{s_0}}{s}, \ldots, \subtree{\treeroot{s_k}}{s})} & \text{induction hypothesis} \\
     & = & \treecommute{\phi(n)(s_0, \ldots, s_k)} \\
     & \to^* & \treecommute{n(s_0, \ldots, s_k)} & \text{decreasing $\phi(n)$'s label} \\
     & = & \subtree{n}{s}
  \end{array}
  \]
  which concludes the proof.
\end{proof}

\begin{example}
  We have $\treecommute{2(9,7(0,4))} \treeemb \treecommute{1(3(8(8(5,1)),9,5(9)),2)}$ since
  \begin{align*}
    &\underline{1(3(8(8(5,1)),9,5(9)),2)} \\
    \to_\trsemb\ &3(\underline{8(8(5,1))},9,5(9)) &&\text{selecting a term} \\
    \to_\trsemb\ &3(8(5,1),\underline{9},5(9)) &&\text{selecting a term} \\
    \to_\trsemb\ &3(8(5,1),\underline{5(9)}) &&\text{removing a subterm} \\
    \to_\trsemb\ &3(8(5,1),9) &&\text{selecting a term} \\
    \commute\ &3(9,8(1,5)) &&\text{commutativity} \\
    \to_\trsemb^+\ &2(9,7(0,4)) &&\text{decreasing labels (multiple times)}
  \end{align*}
\end{example}

\begin{example}[Non-example]
  We do \emph{not} have $\treecommute{1(0,0)} \treeemb \treecommute{1(0(0,0))}$.
\end{example}

The following proposition allows us to infer the existence of an IPO reduction sequence $t \ilpored^* s$ given a tree embedding $s \treeemb t$.

\begin{proposition}\label{ex:treeemb:rpo}
  For $s, t \in \tre{\Sigma}{\avars}$, $s \treeemb t \implies t \ilpored^* s$.
\end{proposition}

\begin{proof}
  By Proposition~\ref{lemma:equivalence:of:embedding:relations} it suffices to show that ${\to_\trsemb^*} \subseteq {\ilpored^*}$. This follows from the easily observed fact that IPO can simulate any ${\to_\trsemb}$ step, i.e., that ${\to_\trsemb} \subseteq {\ilpored^*}$.
\end{proof}

\subsection{Proving acyclicity}\label{sec:star:termination}

Our final obligation concerns the acyclicity of $\ilpored^+$ on star-free trees. The proof that we will give is a microscopic analysis of how the symbols propagate and 
have descendants in a $\ilpored^+$ reduction. Because of this symbol tracing approach, we will have to ``freeze'' the trees $\tre{\Sigma}{\avars}$ of interest, and thus reason on the level of terms $\ter{\Sigma}{\avars}$ rather than trees. The precise nature of this proof reduction is clarified and justified by the following proposition.

\begin{proposition}\label{prop:acyclic:terms:imp:acyclic:trees}
  The relation $\ilpored$ has the following properties:
  \begin{enumerate}[label=(\roman*)]
    \item \label{prop:acyclic:terms:imp:acyclic:trees:postpone}
      On $\ter{\Sigma}{\avars}$, we have ${\commutestep \cdot \ilpored} \subseteq {\ilpored \cdot \commutestep}$.
    \item\label{prop:acyclic:terms:imp:acyclic:trees:imp}
      If for no term $t \in \ter{\Sigma}{\avars}$ we have $t \ (\ilpored^+ \cdot \commute) \ t$, then for no tree $t' \in \tre{\Sigma}{\avars}$ we have $t' \ilpored^+ t'$.
  \end{enumerate}
\end{proposition}
\begin{proof}
\noindent
  \begin{enumerate}[label=(\roman*)]
    \item 
      It is routine to verify that any swap steps can be postponed.
    \item 
      By contradiction. Assume that there is a cycle on trees $t' \ilpored^+ t'$. Then by the definition of a tree one obtains a sequence on terms $\sigma = t' \ (\commutestep^* \cdot \ilpored \cdot \commutestep^*)^+ \ t'$. From $\sigma$ we can construct a sequence $t' \ (\ilpored^+ \cdot \commutestep^*) \ t'$ on terms by repeatedly applying~\ref{prop:acyclic:terms:imp:acyclic:trees:postpone}, contradicting the hypothesis.
      \qedhere
  \end{enumerate}
\end{proof}

  The tracing can be understood intuitively via coloring.
  In order to trace a symbol occurrence~$\rho$ throughout a reduction, one can color $\rho$ using a unique color~$c$.
  Then one can apply the colored variant of the rewrite rules as shown in Figure~\ref{fig:tracing:colours}.
  After each rewrite step, the symbol occurrences with color $c$ are the descendants of $\rho$. This way of proceeding, in which the tracing is  expressed via labelling is called \emph{label-tracing} in~\cite[Section 8.6.4.]{terese}.

\begin{figure}
  {

\begin{tikzpicture}[default,thick,level distance=10mm,sibling distance=8mm,level 2/.style={sibling distance=4mm},inner sep=1.2mm,
  every node/.style={circle},
  p/.style={paint=#1,opacity=.5},
  ny/.style={thick,draw=yellow!70!black,fill=yellow!30,minimum size=6mm},
  no/.style={thick,rectangle,draw=corange!70!black,fill=corange!30,minimum size=5mm},
  ]

  \begin{scope}[yshift=0mm]
    \begin{scope}
\node (f) {$f$}
    child { node (t_1) {$t_1$}
  }
    child { node (ldots) {$\ldots$}
  }
    child { node (t_n) {$t_n$}
  };
\begin{pgfonlayer}{background}
\end{pgfonlayer}

    \end{scope}
    
    \begin{pgfonlayer}{background}
      \node [ny] at (f) {};    
    \end{pgfonlayer}
  
    \begin{scope}[xshift=30mm]
\node (f^star) {$f^\star$}
    child { node (t_1) {$t_1$}
  }
    child { node (ldots) {$\ldots$}
  }
    child { node (t_n) {$t_n$}
  };
\begin{pgfonlayer}{background}
\end{pgfonlayer}

    \end{scope}
  
    \begin{pgfonlayer}{background}
      \node [ny] at (f^star) {};    
    \end{pgfonlayer}
    
    \draw [thick,->] ($(f)+(9mm,0mm)$) to node [above,rectangle] {$\lput$} ($(f^star)+(-9mm,0mm)$);
  \end{scope}

  \begin{scope}[yshift=-20mm]
    \begin{scope}
\node (f^star) {$f^\star$}
    child { node (t_1) {$t_1$}
  }
    child { node (ldots) {$\ldots$}
  }
    child { node (t_n) {$t_n$}
  };
\begin{pgfonlayer}{background}
\end{pgfonlayer}

    \end{scope}
    
    \begin{pgfonlayer}{background}
      \node [ny] at (f^star) {};    
    \end{pgfonlayer}
  
    \begin{scope}[xshift=30mm]
\node (t_i) {$t_i$};
\begin{pgfonlayer}{background}
\end{pgfonlayer}

    \end{scope}

    \draw [thick,->] ($(f^star)+(9mm,0mm)$) to node [above,rectangle] {$\lselect$} ($(t_i)+(-9mm,0mm)$);
  \end{scope}

  \begin{scope}[yshift=-40mm,xshift=0mm]
    \begin{scope}
\node (f^star) {$f^\star$}
    child { node (t_1) {$t_1$}
  }
    child { node (ldots) {$\ldots$}
  }
    child { node (t_n) {$t_n$}
  };
\begin{pgfonlayer}{background}
\end{pgfonlayer}

    \end{scope}
    
    \begin{pgfonlayer}{background}
      \node [ny] at (f^star) {};    
    \end{pgfonlayer}
  
    \begin{scope}[xshift=30mm]
\node (g) {$g$}
    child { node (f1) {$f^\star$}
      child { node (t_1) {$t_1$}
    }
      child { node (ldots) {$\ldots$}
    }
      child { node (t_n) {$t_n$}
    }
  }
    child { node (ldots) {\ldots}
  }
    child { node (f2) {$f^\star$}
      child { node (t_1) {$t_1$}
    }
      child { node (ldots) {$\ldots$}
    }
      child { node (t_n) {$t_n$}
    }
  };
\begin{pgfonlayer}{background}
\end{pgfonlayer}

    \end{scope}
  
    \begin{pgfonlayer}{background}
      \node [ny] at (g) {};    
      \node [ny] at (f1) {};    
      \node [ny] at (f2) {};    
    \end{pgfonlayer}

    \draw [thick,->] ($(f^star)+(9mm,0mm)$) to node [above,rectangle] {$\lcopy$} ($(g)+(-9mm,0mm)$);
  \end{scope}

  \begin{scope}[yshift=-70mm,xshift=-11mm,sibling distance=7mm]
    \begin{scope}
\node (f^star) {$f^\star$}
    child { node (t_1) {$t_1$}
  }
    child { node (ldots) {$\ldots$}
  }
    child { node (t_i-1) {$t_{i-1}$}
  }
    child { node (g) {$g$}
      child { node (s_1) {$s_1$}
    }
      child { node (ldots) {$\ldots$}
    }
      child { node (s_m) {$s_m$}
    }
  }
    child { node (t_i+1) {$t_{i+1}$}
  }
    child { node (ldots) {$\ldots$}
  }
    child { node (t_n) {$t_n$}
  };
\begin{pgfonlayer}{background}
\end{pgfonlayer}

    \end{scope}
    
    \begin{pgfonlayer}{background}
      \node [ny] at (f^star) {};    
      \node [no] at (g) {};    
    \end{pgfonlayer}
  
    \begin{scope}[xshift=55mm,sibling distance=7mm]
\node (f) {$f$}
    child { node (t_1) {$t_1$}
  }
    child { node (ldots) {$\ldots$}
  }
    child { node (t_i-1) {$t_{i-1}$}
  }
    child { node (g1) {$g^\star$}
      child { node (s_1) {$s_1$}
    }
      child { node (ldots) {$\ldots$}
    }
      child { node (s_m) {$s_m$}
    }
  }
    child { node (ldots) {$\ldots$}
  }
    child { node (g2) {$g^\star$}
      child { node (s_1) {$s_1$}
    }
      child { node (ldots) {$\ldots$}
    }
      child { node (s_m) {$s_m$}
    }
  }
    child { node (t_i+1) {$t_{i+1}$}
  }
    child { node (ldots) {$\ldots$}
  }
    child { node (t_n) {$t_n$}
  };
\begin{pgfonlayer}{background}
\end{pgfonlayer}

    \end{scope}
  
    \begin{pgfonlayer}{background}
      \node [ny] at (f) {};    
      \node [no] at (g1) {};    
      \node [no] at (g2) {};    
    \end{pgfonlayer}

    \draw [thick,->] ($(f^star)+(15mm,0mm)$) to node [above,rectangle] {$\ldown$} ($(f)+(-15mm,0mm)$);
  \end{scope}
\end{tikzpicture}
}  
  \caption{Tracing via colors.
    The symbol occurrences in the right-hand side are descendants of those symbol occurrences in the left-hand side that have the same color.
    Here red (square) and yellow (round) can be arbitrary colors (they can also be the same color).
    The subterms $s_1,\ldots,s_m,t_1,\ldots,t_n$ can contain colored symbol occurrences as well.}
  \label{fig:tracing:colours}
\end{figure}
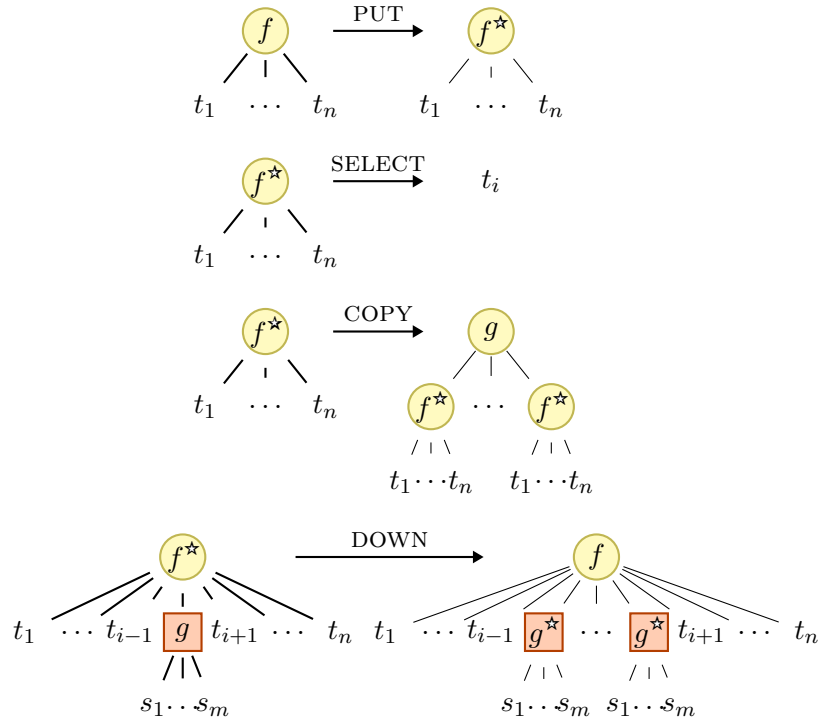

Here, we define the descendant relation without the use of colors as follows.

\begin{definition}[Tracing]\label{def:star:tracing}
  We define in a step $s \ilpored t$ ($s,t \in \terstar{\Sigma}{\avars}$)
  a \emph{descendant relation} between the occurrences of constant and function symbols of $s$ and $t$.
  We will disregard the markers $\star$.
  The definition follows the clauses of the definition of $\ilpored$.
  \begin{enumerate}
    \item
      $s = f(\vec{u}) \ilpored_\lput f^\star(\vec{u}) = t$.
      Every symbol occurrence in $s$ has as unique descendant ``the same'' symbol occurrence in $t$. (In particular, $f^\star$ descends from $f$.)
    \item 
      $s = f^\star(\vec{u}, v, \vec{w}) \ilpored_\lselect v = t$.
      Every symbol occurrence in $v$ (subterm of~$s$) has as unique descendant ``the same'' symbol occurrence in $v = t$. 
      The other symbol occurrences of $s$ have no descendant.
     \item\label{def:star:tracing:copy}
      $s = f^\star(\vec{u}) \ilpored_\lcopy g(s,\ldots,s) = t$.
      \begin{enumerate}[label=(\alph*)]
        \item 
          Every symbol occurrence in $s$ has as descendants the corresponding symbols in each copy of $s$ in $t$.
        \item \label{def:star:tracing:copy:special}
          In addition, the head symbol $f^\star$ of $s$ has the head symbol $g$ of $t$ as descendant. 
          In this case $g$ is called the \emph{special} descendant of $f^\star$.
      \end{enumerate} 
    \item
      $s = f^\star(\vec{u},v, \vec{w}) \ilpored_\ldown f(\vec{u}, v^\star, \ldots, v^\star, \vec{w}) = t$.
      Every symbol occurrence in $s$ has as unique descendant ``the same'' symbol in $t$.
    \item
      $s = \cxtfill{C}{t_1} \ilpored \cxtfill{C}{t_2} = t$ with $t_1 \ilpored t_2$.
      The descendants of symbol occurrences in $t_1$ are given by the clauses above;
      each symbol occurrence in the context $\cxtfill{C}{}$ has as unique descendant ``the same'' symbol in $t$.
  \end{enumerate}
  In a sequence $t_0 \ilpored t_1 \ilpored t_2 \ilpored \cdots \ilpored t_n$ of several steps ($n \geq 2$) 
  we define the descendant relation between symbol occurrences of $t_0$ and those of $t_n$ in the obvious way, namely by relational composition. Moreover, if $p$ is a descendant of $q$, then we call $q$ an \emph{ancestor} of $p$.
\end{definition}

The descendant relation thus defined has the following three properties that are easily checked.

\begin{lemma}[Unique ancestors]
  Let $s \ilpored^+ t$, and let $q$ be a symbol occurrence in $t$. 
  Then $q$ has a unique ancestor in $s$.
\end{lemma}

\begin{lemma}[Descendants are smaller]\label{lem:descendants_are_smaller}
  Let $p$ and $q$ be symbol occurrences such that $q$ descends from $p$.
  Then $p \geq q$ (in the ordering on $S$). 
  If one or more steps in the descendant sequence from $p$ to $q$ are special, i.e., are of type Definition~\ref{def:star:tracing}.\ref{def:star:tracing:copy}\ref{def:star:tracing:copy:special},
  then moreover $p > q$.
\end{lemma}

  \begin{figure}[h!]
    \centering
    \begin{tikzpicture}[scale=0.7, default,thick,arrow/.style={->,arrows={-Triangle[angle=60:10pt,black,fill=white, scale=0.8]}}]
      \begin{scope}[xshift=0mm]
        \draw (0,0) to ++(20mm,-40mm) to ++(-40mm,0mm) to cycle;
        \node (t) [anchor=south] at (0,0) {$t$};
        \draw [fill=cblue!20] (5mm,-25mm) to ++(5mm,-15mm) to ++(-10mm,0mm) to cycle;
        \node (p) [anchor=east] at (5mm,-25mm) {$p$};
      \end{scope}
      \begin{scope}[xshift=62mm]
        \draw (0,0) to ++(20mm,-40mm) to ++(-40mm,0mm) to cycle;
        \node (t2) [anchor=south] at (0,0) {$t$};
        \draw [fill=cblue!20] (-1mm,-20mm) to ++(10mm,-20mm) to ++(-20mm,0mm) to cycle;
        \node (q) [anchor=west] at (-1mm,-20mm) {$q$};
      \end{scope}
      \draw [arrow] ($(t)+(4mm,0mm)$) to ++(10mm,0mm);
      \draw [arrow] ($(t)+(15mm,0mm)$) to ++(10mm,0mm);
      \draw [arrow] ($(t)+(26mm,0mm)$) to ++(10mm,0mm);
      \draw [arrow] ($(t)+(37mm,0mm)$) to ++(10mm,0mm);
      \draw [arrow] ($(t)+(48mm,0mm)$) to ++(10mm,0mm);

      \draw [ultra thick,cblue!30] (p) to[out=0,in=180] ++(30mm,10mm) to[out=0,in=180] (q);
      \begin{scope}[shorten <= 1mm, shorten >= 1mm]
      \coordinate (p) at ($(p)+(4mm,-2mm)$);
      \coordinate (q) at ($(q)+(-4mm,-2mm)$);
      \draw [arrow] ($(p)!0.0!(q)$) to ($(p)!0.25!(q)$);
      \draw [arrow] ($(p)!0.25!(q)$) to ($(p)!0.5!(q)$);
      \draw [arrow] ($(p)!0.5!(q)$) to ($(p)!.75!(q)$);
      \draw [arrow] ($(p)!0.75!(q)$) to ($(p)!1!(q)$);
      \end{scope}
    \end{tikzpicture}
    \caption{Illustration of Lemma~\ref{lem:star:subtrees}.}
    \label{fig:lem:star:subtrees}
  \end{figure}
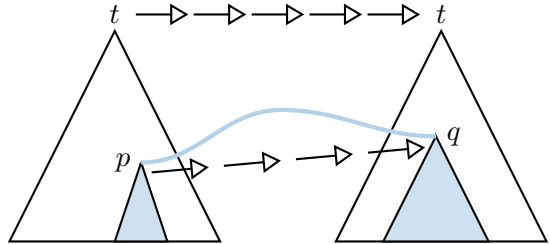

The following lemma is illustrated in Figure~\ref{fig:lem:star:subtrees}.

\begin{lemma}[Reductions on subterms]\label{lem:star:subtrees}
  Let $\alpha : s \ilpored^+ t$ be a nonempty sequence.
  Let $p$ be a symbol occurrence in $s$ and $q$ one of its descendants in $t$.
  Let $\subtermat{s}{p}$ be the subterm of $s$ rooted at symbol occurrence $p$,
  and $\subtermat{t}{q}$ the subterm of $t$ rooted at symbol occurrence $q$.
  Then $\subtermat{s}{p} \ilpored^* \subtermat{t}{q}$, 
  in a number of $\ilpored$-steps less or equal to the number of steps in the sequence $\alpha$.
\end{lemma}

\begin{proof}
  This property is easily proved from the particular case of one step $s \ilpored t$, where for $p$ in $s$ and $q$ in $t$ such that $q$ descends from $p$, we have $\subtermat{s}{p} \ilpored^= \subtermat{t}{q}$.
\end{proof}

We are now ready to prove acyclicity on terms.

\begin{lemma}[Acyclicity on terms]\label{prop:rpo:acyclic}
For no term $t \in \ter{\Sigma}{\avars}$,  $t \ (\ilpored^+ \cdot \commute) \ t$
\end{lemma}
\begin{proof}
  For a contradiction, 
  suppose there exists a cycle $t \ (\ilpored^+ \cdot \commute) \ t$ on terms.
  We call such a term $t$ \emph{cyclic}. 
  Let $t$ be a smallest cyclic term (with respect to the number of symbols).
  We consider a shortest cycle of $t$ (with respect to the number of steps):
  \begin{align*}
    \text{$\mathcal{C}$: $t = f(\vec{u}) \ilpored^+ t' = f(\vec{v}) \commute t$,}
  \end{align*}
  \begin{enumerate}
  \item \label{item:head_t'_descends_from_head_t}
    Our first observation on the cycle $\mathcal{C}$ is that 
    the head symbol of $t'$ descends from the head symbol of $t$.
    
    For suppose that the head symbol $f$ of $t'$ descends from a symbol occurrence in $t$ at position $q \ne \posemp$. Then, by Lemma~\ref{lem:star:subtrees}, we have $\subtermat{t}{q} \ilpored^* t'$.  Since $t \commute t'$, there is a $q'$ in $t'$
  such that $\subtermat{t'}{q'} \commute \subtermat{t}{q}$, $q'$ not at the root.
  (That is, $\subtermat{t'}{q'}$ is a proper subterm of $t'$.)
  Therefore $t' \ilpored^+ \subtermat{t'}{q'}$ by, e.g., a $(\ilpored_\lput \cdot \ilpored_\lselect)^+$ sequence at the root. But then we obtain the cycle
  \[
  \subtermat{t}{q} \ilpored^* t' \ilpored^+ \subtermat{t'}{q'} \commute \subtermat{t}{q}
  \]
  which contradicts that $t$ that is the smallest cyclic term (Figure~\ref{prop:rpo:acyclic:item1}).

  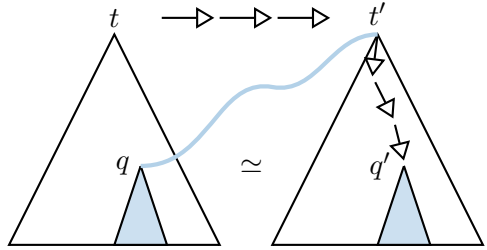
\begin{figure}[h!]
    \centering
    \begin{tikzpicture}[scale=0.7,default,thick,arrow/.style={->,arrows={-Triangle[angle=60:10pt,black,fill=white,scale=0.8]}}]
      \begin{scope}[xshift=0mm]
        \draw (0,0) to ++(20mm,-40mm) to ++(-40mm,0mm) to cycle;
        \draw [fill=cblue!20] (5mm,-25mm) to ++(5mm,-15mm) to ++(-10mm,0mm) to cycle;
        \node (t) [anchor=south] at (0,0) {$t$};
        \node (q) [anchor=east] at (5mm,-25mm) {$q$};
      \end{scope}
      \begin{scope}[xshift=50mm]
        \draw (0,0) to ++(20mm,-40mm) to ++(-40mm,0mm) to cycle;
        \draw [fill=cblue!20] (5mm,-25mm) to ++(5mm,-15mm) to ++(-10mm,0mm) to cycle;
        \node (t2) [anchor=south] at (0,0) {$t'$};
        \draw [arrow] (0,0) to node (c) [at end] {} ++(-1mm,-8mm);
        \draw [arrow] (c) to node (c) [at end] {} ++(4mm,-8mm);
        \draw [arrow] (c) to node (c) [at end] {} ++(2mm,-8mm);
        \node (q2) [anchor=east] at (5mm,-25mm) {$q'$};
      \end{scope}
      \begin{scope}[xshift=100mm]
      \end{scope}
      \draw [arrow] ($(t)+(9mm,0mm)$) to ++(10mm,0mm);
      \draw [arrow] ($(t)+(20mm,0mm)$) to ++(10mm,0mm);
      \draw [arrow] ($(t)+(31mm,0mm)$) to ++(10mm,0mm);
      \node at ($(q)!.5!(q2)$) {$\commute$};
    
      \draw [ultra thick,cblue!30] (5mm,-25mm) to[out=0,in=-130] ++(15mm,10mm) to[out=-130+180,in=170] ++(10mm,5mm) to[out=170+180,in=180] ++(20mm,10mm);
    \end{tikzpicture}
    \caption{Illustration of item (1) in the proof of Proposition~\ref{prop:rpo:acyclic}.}\label{prop:rpo:acyclic:item1}
  \end{figure}

  \item\label{item:no_select_or_copy_at_root} The second observation is that in $\mathcal{C}$ there do not occur
  steps of type $\lselect$ or $\lcopy$ at the root.
  
  For suppose there is such a first step in $\mathcal{C}$. This cannot be a $\lselect$ step, since otherwise the head symbol of $t$ would not have a descendant in $t'$, contradicting the previous item. So $\mathcal{C}$ must be of the form
  \[
  t = f(\vec{u}) \ilpored^+_{\alpha} f^\star(\vec{w}) = s
  \ilpored_{\lcopy} g(s,\ldots,s) \ilpored^*_{\beta} t' \commute t
  \]
  for some $g < f$, where in part $\alpha$ no $\lselect$ or $\lcopy$ steps occur at the root. (Observe that $f(\vec{u}) \ilpored^+_{\alpha} f^\star(\vec{w})$ is nonempty because $f$ must become starred for a $\lcopy$ step to occur.)
  
  Now the only descendants in $g(s,\ldots,s)$ of the head symbol $f$ of $t$ are the head symbol $g$ and the head symbols of the arguments $s,\ldots,s$.
  By item~\eqref{item:head_t'_descends_from_head_t} and the fact that the descendant relation is defined through relational composition, one of these symbols must be the ancestor of the head symbol $f$ of $t'$. By Lemma \ref{lem:descendants_are_smaller} and $g < f$, this ancestor cannot be $g$. Thus, it must be the head symbol $f$ of one of the arguments $s,\ldots, s$.
  By Lemma \ref{lem:star:subtrees} we then obtain a sequence $s \ilpored^*_\gamma t'$ of length not
  exceeding the length of part $\beta$. Omitting the first $\lcopy$ step, it follows that 
  \[
  t = f(\vec{u}) \ilpored^+_{\alpha} f^\star(\vec{w}) = s
  \ilpored^*_\gamma t' \commute t
  \]
  is a cycle on $t$ shorter than $\mathcal{C}$, contradicting $\mathcal{C}$'s
  minimality.

  \item
  Let $\ldown(T, k)$ denote a $\ldown$ step \emph{at a root in $\mathcal{C}$} which creates $k$ starred copies of a subtree $T$.
  Our final observation is that any step at a root in $\mathcal{C}$ is either a $\lput$ step or a $\ldown(k)$ step with $k > 1$, and that there is at least one such a $\ldown$ step. From this it follows that the root of $t'$ has more subterms than the root of $t$, contradicting $t \commute t'$, and concluding the proof.
  
  The final observation follows from the following three items:
  \begin{enumerate}
      \item At least one $\ldown$ step must occur in $\mathcal{C}$. Otherwise, it follows from item~\eqref{item:no_select_or_copy_at_root} that all steps in $\mathcal{C}$ happen below the root, implying the existence of a cyclic term smaller than $t$.
     \item No $\ldown(T, 1)$ step takes place. For the step and its corresponding $\lput$ step at the root can be replaced with a single $\lput$ step on the root of $T$, yielding a cycle smaller than $\mathcal{C}$.
    \item No $\ldown(T,0)$ step takes place. For a contradiction, suppose otherwise. Since no $\lselect$ and $\lcopy$ steps take place at the root, $T$ was either originally under the head symbol of $t$, or it arose due to some earlier $\ldown(T', k)$ step with $k > 1$ and $T$ a reduct of $(T')^\star$. In either case, we know that no local rewrite steps have since taken place in $T$, since omitting these steps would yield a cycle shorter than $\mathcal{C}$.
    
    In the first case, $T$ can be pruned from the initial term $t$, yielding a smaller cyclic term. More formally, we have
    \[
    t 
    = 
    f(\vec{w},T,\vec{x})
    \ilpored^+
    f^\star(\vec{y},T,\vec{z})
    \ilpored_{\ldown(T,0)} 
    f(\vec{y}, \vec{z})
    \ilpored^*
    t'
    \commute
    t \text{,}
    \]
   from which we can construct a sequence
    \[
    f(\vec{w}, \vec{x})
    \ilpored^*
    f(\vec{y},\vec{z})
    \ilpored^*
    t'
    \commute
    t 
    =
    f(\vec{w},T,\vec{x})
    \ilpored^+
    f(\vec{w}, \vec{x})
    \]
  with $f(\vec{w}, \vec{x})$ the cyclic term smaller than $t$ (postponing the ${\commute}$ step).
   
   In the second case, we can make one less copy of $T'$ in the earlier $\ldown(T', k)$ step, allowing us to construct a cycle shorter than $\mathcal{C}$ by omitting the $\ldown(T,0)$ step. More formally, we have
   \[
   \begin{array}{llll}
    t 
    \
    (\ilpored^+ \cdot \ilpored_{\ldown(T',k)})
    \
    f(\vec{w}, (T')^\star, \vec{x})
    \ilpored^+_\alpha
    f^\star(\vec{y}, T, \vec{z})
    & \ilpored_{\ldown(T,0)}
    f(\vec{y}, \vec{z}) \\
    &\ilpored^* t' \commute t
    \end{array}
    \]
    from which we can construct a shorter cycle
       \[
    t 
    \
    (\ilpored^+ \cdot \ilpored_{\ldown(T', k-1)})
    \
    f(\vec{w}, \vec{x})
    \ilpored^*_\beta
    f(\vec{y}, \vec{z})
    \ilpored^* t' \commute t
    \]
   where part $\beta$ is shorter than $\alpha$.\qedhere
\end{enumerate}
  \end{enumerate}
\end{proof}

This ends the proof of acyclicity of the iterative path ordering on terms $\ter{\Sigma}{\avars}$.

\begin{corollary}[Acyclicity on trees]\label{lem:rpo:acyclic}
  For no tree $t \in \tre{\Sigma}{\avars}$ we have $t \ilpored^+ t$.
\end{corollary}
\begin{proof}
Immediate from Proposition~\ref{prop:acyclic:terms:imp:acyclic:trees}\ref{prop:acyclic:terms:imp:acyclic:trees:imp} and Lemma~\ref{prop:rpo:acyclic}.
\end{proof}

We now have all the ingredients required for proving Theorem~\ref{thm:rpo:tree:sn}, the claim that $\ilpored^+$ restricted to unstarred trees is SN.

\begin{proof}[Proof of Theorem~\ref{thm:rpo:tree:sn}] 
  Follows from the basic Proposition~\ref{basic}, Kruskal's tree theorem (Theorem~\ref{thm:kruskal}), Proposition~\ref{ex:treeemb:rpo} and Corollary~\ref{lem:rpo:acyclic}.
\end{proof}

\section{The Second Proof: a Magical Proof}\label{sec:buchholz}

We now give an alternative proof of Theorem~\ref{thm:rpo:tree:sn}, by employing an ingenious proof technique due to Buchholz~\cite{buchholz1995proof} which uses neither Kruskal's tree theorem nor acyclicity.
Throughout this section, $(\Sigma,<)$ is a well-founded partial order.\footnote{
  Strictly speaking even transitivity is not needed, an arbitrary well-founded relation suffices. 
  However, without loss of generality, we may assume that $<$ is transitive, for otherwise we can work with $<^+$.
}

As we have seen, the star game introduced and studied so far 
has the peculiarity that it is terminating when restricted to 
the intended, proper objects, 
being finite trees with nodes labeled with natural numbers or ordinals, 
but without presence of the stars as ``control agents'’. 
On the auxiliary trees that are equipped with the control stars, 
we do not have termination, as there may be cycles and spiraling reductions,
generating a proper context of the starting tree.

An interesting subsequent step was made by Vincent van Oostrom in~\cite{klop2006iterative}, 
who noticed that by making the control stars \emph{quantitative}, 
as if having a certain amount of working energy,
conceived as a natural number or ordinal, this peculiarity can be made to disappear, 
and the auxiliary objects are themselves already terminating. 
Moreover, the two reductions, with simple stars or with their numerical refinements, 
are coinciding on the proper objects, trees without control symbols.
Figure~\ref{star-labels} presents a helicopter view of this situation.\footnote{%
	In~\cite{klop2006iterative} the present second author together with his co-authors
	chose to focus on the iterative \emph{lexicographic} path order among
	the material in draft form available at the time. That Buchholz' proof
	method could also be adapted to prove termination of related orders such
	as the iterative \emph{recursive} path order was not fleshed out and
	ended up only as a claim in the conclusion of~\cite{klop2006iterative}.
	We prove the claim applies to IPO (Definition~\ref{def:star:game}) as employed here.
}

We now state the rules of the number star game.
\begin{definition}[Energized signature]
Let $f^n$ be a fresh symbol for every $f \in \Sigma$ and $n \in \nat$.
We define 
$
\Sigma^\omega = \set{ f^n \mid f \in \Sigma,\, n \in \nat }$
as the \emph{energized signature}.
\end{definition}

\newcommand{\teromega}[2]{\ter{#1}{#2}^\omega}
\newcommand{\treomega}[2]{\tre{#1}{#2}^\omega}
\begin{notation}
We abbreviate %
$\tre{\Sigma \cup \Sigma^\omega}{\avars}$
as $\treomega{\Sigma}{\avars}$.
\end{notation}

\begin{definition}[Stars with specified energy]\label{def:star:game:nat}%
  On $\treomega{\Sigma}{\avars}$ we define the rewriting system $\trsomega$,
  which induces the reduction relation $\ilporedomega$, as follows:
  \begin{gather*}
    \begin{aligned}
      (\lput) && f(\vec{s}) &\ilporedomega f^n(\vec{s})\\
      (\lselect) && f^{n+1}(\vec{s},t,\vec{u}) &\ilporedomega t  \\
      (\lcopy) && f^{n+1}(\vec{s}) &\ilporedomega g(\underbrace{f^n(\vec{s}),\ldots,f^n(\vec{s})}_{\text{$0$ or more}}) && \text{if $f > g$} \\
      (\ldown) && f^{n+1}( \vec{s},t,\vec{u} ) &\ilporedomega f( \vec{s}, \underbrace{t^n, \ldots, t^n}_{\text{$0$ or more}}, \vec{u} )
    \end{aligned}
  \end{gather*}
  for every every $n \in \nat$ and $f,g \in \Sigma$.
\end{definition}

A star $\star$ can be seen as an unspecified amount of ``energy''
while a label $n$ is a precise quantity of energy. 
We can lift every finite reduction in $\trsstar$ to one in $\trsomega$
by replacing the stars by energy quanta 
that are large enough to make the journey all the way. 
As the energy decreases by at most $1$ in every $\ilporedomega$ step, 
the length of the (remaining) reduction bounds the amount of energy that we need.

\begin{proposition} \label{lem:ilex:llex}
  The relations
  $\ilpored^+$ and
  $\ilporedomega^+$ coincide
  on unmarked trees $\tre{\Sigma}{\avars}$. \qed
\end{proposition}

\begin{example}
  We reconsider Example~\ref{ex:irpo_application}
  with the single rule $1(0(x)) \to 0(0(1(x)))$ and the order $1 > 0$ on the signature.
  In $\trsomega$ we have
  \[
  \begin{array}{llll}
    \treecommute{1(0(x))} & \ilporedomega_\lput & \treecommute{1^4(0(x))} \\
    & \ilporedomega_\lcopy & \treecommute{0(1^3(0(x)))} \\
    & \ilporedomega_\lcopy & \treecommute{0(0(1^2(0(x))))} \\
    &\ilporedomega_\ldown & \treecommute{0(0(1(0^1(x))))} \\
    &\ilporedomega_\lselect & \treecommute{0(0(1(x)))}
  \end{array}
  \]
  Note that the energy decreases in each step.
\end{example}

We now prove the main theorem, termination of $\ilporedomega$, by employing a proof technique due to Buchholz~\cite{buchholz1995proof}. As mentioned before, the proof technique was adapted by Vincent van Oostrom, and is included in~\cite{klop2006iterative}. In the following proof we closely follow the exposition given by van Oostrom 
in~\cite{klop2006iterative}.

\begin{proposition} \label{prop:ilpoomega:sn}
  $\ilporedomega$ is terminating.
\end{proposition}
\newcommand{\wflex}{\hookrightarrow}
\begin{proof}
  We write $\funap{f^{\omega}}{\vec{t}}$ to denote the tree $\funap{f}{\vec{t}}$ 
  for an \emph{unmarked} symbol $f$.
  This allows us to write \emph{any} tree from $\treomega{\Sigma}{\avars}$ uniquely in the form $\funap{f^{\alpha}}{\vec{t}}$ 
  for some unmarked symbol $f$, vector of trees $\vec{t}$ 
  and ordinal $\alpha \le \omega$.
  In the crucial induction in this proof we will make use of the fact 
  that in the ordering of the ordinals we have $\omega > n$ for each natural number $n$. 

  We prove by induction on the construction of trees that any $\omega$-marked tree is terminating.
  To be precise, we show that any $\omega$-marked tree $\funap{f^{\alpha}}{\vec{t}}$ is terminating
  under the assumption (the induction hypothesis) that its arguments $\vec{t}$ are terminating:
  \begin{center}
    $t_1,\ldots,t_n$ terminating $\quad\implies\quad$ $\funap{f^\alpha}{t_1,\ldots,t_n}$ terminating
  \end{center}  
  for every $f \in \Sigma$ and $\alpha \le \omega$.
  We prove this implication by induction on the triple
  \[
  \triple{f}{\mset{t_1,\ldots,t_n}}{\alpha}
  \]
  in the lexicographic order with respect to 
  \begin{enumerate}
    \item $>$ on the symbols of $\Sigma$, 
    \item the multiset extension $\ilporedomegamulti$ of  $\ilporedomega$, and
    \item $>$ on the markers in $\nat \cup \set{\omega}$.
  \end{enumerate}
  Note that $\ilporedomegamulti$ is well-founded on multisets that contain only terminating terms.

  Clearly, the tree $\funap{f^{\alpha}}{\vec{t}}$ is terminating if all its one-step reducts are.
  We prove that the one-step reducts are terminating by distinguishing cases on the type of the reduction step:
  \begin{enumerate}
  \item
    If the step occurs at the root, we perform a further case analysis on the rule applied:
    \begin{list}{}{}
    \item[($\lput$)]
    The first and second component of the triple remain the same. The third~component decreases from $\omega$ to some $m \in \nat$. Thus by the IH, the reduct is terminating.
    \item[($\lselect$)]
      The result follows by the termination assumption for the terms $\vec{t}$.
    \item[($\lcopy$)]
      Then $\alpha$ is of the form $m+1$ for some $m \in \nat$,
      and the reduct has shape $\funap{g^{\omega}}{\funap{f^{\,m}}{\vec{t}}, \ldots, \funap{f^{\,m}}{\vec{t}}}$ for some $g$ such that $f > g$.
      By the IH for the third component, each of the $\funap{f^{\,m}}{\vec{t}}$ is terminating.
      Hence, by the IH for the first component,
      the reduct is terminating.
    \item[($\ldown$)]
      Then $\alpha$ is of the form $m+1$ for some $m \in \nat$,
      and the reduct is of the form 
      \begin{align*}
      f^\omega( t_1,\ldots,t_{i-1}, \underbrace{t_i^m, \ldots, t_i^m}_{\text{$0$ or more}}, t_{i+1},\ldots,t_k )
      \end{align*}
      with $t_i = \funap{g^{\omega}}{\vec{s}}$ 
      for some trees $\vec{s}$,
      and $t_i^m = \funap{g^{m}}{\vec{s}}$.
      Each $t_1,\ldots,t_k$ is terminating by assumption, and $t_i^m$ is terminating since it is a reduct of $t_i$.
      Hence, the reduct is terminating by the IH for the second component.
    \end{list}
  \item
    If the step is a non-head step,
    then it rewrites some subterm, and
    the result follows by the IH for the second component. \qedhere
  \end{enumerate}
\end{proof}

\section{Hydras}\label{sec:hydras}

A (mathematical) Hydra is a tree, in which the leaves resemble the Hydra's heads, and the root resembles its body. The objective is to slay the Hydra by chopping off its heads (i.e., removing its leaves) one by one. Like its mythological counterpart, the Hydra may -- subject to certain conditions -- grow new heads in response to each chop. \emph{Will the Hydra eventually be slain, no matter which strategy we choose?}

\newcommand{\kpto}{\to_\mathit{KP}}
Multiple versions of the Hydra exist. One Hydra of particular interest is the \emph{Kirby-Paris} (\emph{KP}) \emph{Hydra}~\cite{kirby1982accessible}.
It was used by Kirby and Paris to derive a mathematical truth that cannot be proved in Peano Arithmetic (namely, the fact there is no losing strategy against this Hydra). Such independence results started with G\"odel's famous theorem; but the endeavour later on was to find independence results with a more ``mathematical" content, referring to mathematical objects such as words and trees rather than logical proof systems. (Of course, for us, the latter are also mathematical objects, but maybe not in the eyes of every mathematician.)

\newcommand{\kproot}{\textsc{KP}_1}
\newcommand{\kpbody}{\textsc{KP}_2}
\begin{definition}[Kirby-Paris Hydra]
A \emph{Kirby-Paris} (\emph{KP}) \emph{Hydra} is a finite tree with nodes labeled from $S = \{ 0, \dagger \}$. The label $\dagger$ labels only the root. Its rewrite relation $\kpto$ is defined by the TRS consisting of rule schemas
\begin{gather*}
    \begin{aligned}
      (\kproot) && \dagger(\vec{t}, \mathbf{0})  &\to \dagger(\vec{t})\\
      (\kpbody) && f(\vec{t_1}, 0(\vec{t_2}, \mathbf{0})) &\to f(\vec{t_1}, \underbrace{0(\vec{t_2}), \ldots, 0(\vec{t_2})}_{\text{$k \in \mathbb{N}$ times}})
    \end{aligned}
  \end{gather*}
for any $f \in S$. In these rules, a boldfaced $\mathbf{0}$ highlights the head $0$ that is ``chopped off'' by the rule.
\end{definition}

\begin{remark}
The Kirby-Paris Hydra is usually described as an unlabeled Hydra -- or equivalently -- a Hydra labeled from a singleton label set. In this modeling, rule $\kproot$ is of the form $0(\vec{t}, \mathbf{0}) \to 0(\vec{t})$, and its applicability must be explicitly restricted to the root.
\end{remark}

\begin{example}[Slaying a KP Hydra]
\label{ex:kp-hydra}
One battle with the KP Hydra $\dagger(0(0,0))$ may progress as follows:
\[
\begin{array}{llll}
\dagger(0(0,\mathbf{0})) & \kpto & \dagger(0(\mathbf{0}), 0(0)) \\
& \kpto & \dagger(0,\mathbf{0},0(0)) \\
& \kpto & \dagger(\mathbf{0},0(0)) \\
& \kpto & \dagger(0(\mathbf{0})) \\ 
& \kpto & \dagger(0, \mathbf{0}) \\ 
& \kpto & \dagger(\mathbf{0}) \\ 
& \kpto & \dagger \\ 
\end{array}
\]
\end{example}

The KP Hydra of Example~\ref{ex:kp-hydra} started out quite small, and made only two copies after each cut not directly below the root. But it is easy to see that the Hydra can grow explosively large even for small starting examples.

Similar to Example~\ref{ex:irpo_application}, we can use $\trsstar$ to prove that the Kirby-Paris Hydra is eventually slain.
\begin{theorem}\label{thm:kp:can:be:slaim}
Any Kirby-Paris Hydra is eventually slain, no matter the strategy. In other words, $\kpto$ is terminating.
\end{theorem}
\begin{proof}
For $\kproot$ we have:
\[
\dagger(\vec{t}, 0) \ilpored_\lput
\dagger^\star(\vec{t}, 0)
\ilpored_\ldown
\dagger(\vec{t})
\]
and for rule $\kpbody$ we have (see also Figure~\ref{fig:kp:simulation}):
\[
\begin{array}{lcll}
x(\vec{t_1}, 0(\vec{t_2}, 0))
\ilpored_\lput
x^\star(\vec{t_1}, 0(\vec{t_2}, 0))
&
\ilpored_\ldown
&
x(\vec{t_1}, \underbrace{0^\star(\vec{t_2}, 0), \ldots, 0^\star(\vec{t_2}, 0)}_{\text{$k \in \mathbb{N}$ times}}) \\
&
\ilpored_\ldown^*
&
x(\vec{t_1}, \underbrace{0(\vec{t_2}), \ldots, 0(\vec{t_2})}_{\text{$k \in \mathbb{N}$ times}})
\end{array}
\]
Hence ${\kpto} \subseteq {\ilpored^+}$, so that termination of $\kpto$ follows from Corollary~\ref{cor:termination:by:simulation}.
\end{proof}

{%
\newcommand{\anode}[5][]{
  \node (#4) [n,#2] {#1};
  \node at (#4.north east) [anchor=south west,outer sep=0,inner sep=0,rectangle,#3] {#5};
}%
\newcommand{\atree}[2]{
  \begin{pgfonlayer}{background}
    \draw [draw=black,fill=cblue!30] (#1.center) to ++(-2mm,-10mm) to ++(-10mm,0mm) to cycle; 
    \node at ($(#1)+(-4mm,-7mm)$) {$#2$};
  \end{pgfonlayer}
}%
\newcommand{\btree}[2]{
  \begin{pgfonlayer}{background}
    \draw [draw=black,fill=cblue!30] (#1.center) to ++(0mm,-10mm) to ++(-8mm,0mm) to cycle; 
    \node at ($(#1)+(-2mm,-7mm)$) {$#2$};
  \end{pgfonlayer}
}%
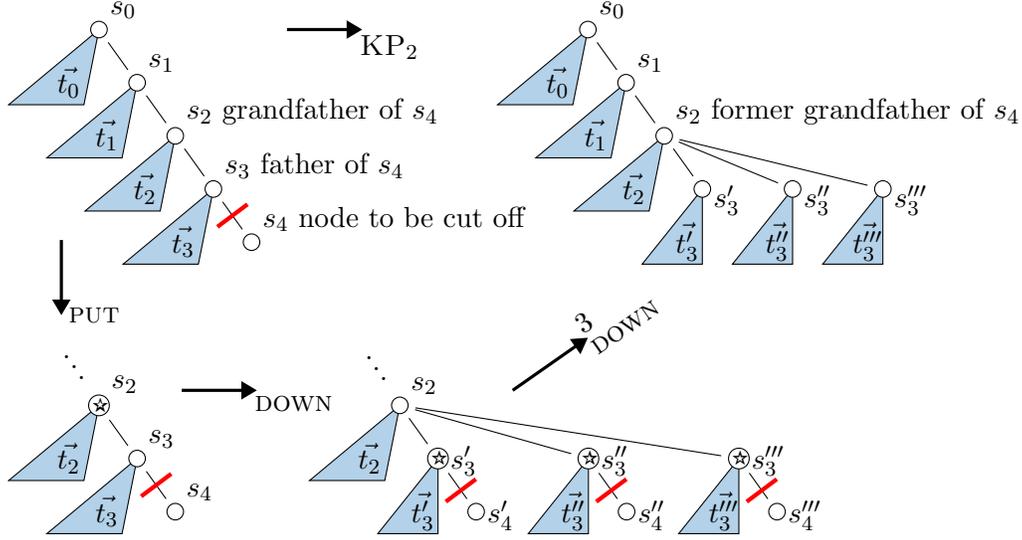
\begin{figure}
  \centering
    
  \begin{tikzpicture}[default,n/.style={smallCircle,fill=white,draw=black,minimum size=2.2mm},thin]
    \begin{scope}
      \anode{}{}{s_0}{$s_0$}
      \atree{s_0}{\vec{t_0}}
      \anode{below right of=s_0,xshift=-2mm}{}{s_1}{$s_1$} \draw (s_0) to (s_1);
      \atree{s_1}{\vec{t_1}}
      \anode{below right of=s_1,xshift=-2mm}{}{s_2}{$s_2$ grandfather of $s_4$} \draw (s_1) to (s_2);
      \atree{s_2}{\vec{t_2}}
      \anode{below right of=s_2,xshift=-2mm}{}{s_3}{$s_3$ father of $s_4$} \draw (s_2) to (s_3);
      \atree{s_3}{\vec{t_3}}
      \anode{below right of=s_3,xshift=-2mm}{}{s_4}{$s_4$ node to be cut off} \draw (s_3) to (s_4);
      \draw [ultra thick,red] ($(s_3)!.5!(s_4) + (-2mm,-1.5mm)$) to ++(4mm,3mm);
    \end{scope}
  
    \begin{scope}[xshift=65mm]
      \draw [very thick,->] (-40mm,0mm) to node [at end,below,anchor=north west,rectangle,inner sep=.5mm,xshift=-1mm] {$\kpbody$} ++(10mm,0mm);
  
      \anode{}{}{s_0}{$s_0$}
      \atree{s_0}{\vec{t_0}}
      \anode{below right of=s_0,xshift=-2mm}{}{s_1}{$s_1$} \draw (s_0) to (s_1);
      \atree{s_1}{\vec{t_1}}
      \anode{below right of=s_1,xshift=-2mm}{}{s_2}{$s_2$ former grandfather of $s_4$} \draw (s_1) to (s_2);
      \atree{s_2}{\vec{t_2}}
  
      \anode{below right of=s_2,xshift=-2mm}{anchor=north west,yshift=-1mm}{s_3'}{$s_3'$} \draw (s_2) to (s_3');
      \btree{s_3'}{\vec{t_3'}}
      \anode{right of=s_3',node distance=12mm}{anchor=north west,yshift=-1mm}{s_3''}{$s_3''$} \draw (s_2) to (s_3'');
      \btree{s_3''}{\vec{t_3''}}
      \anode{right of=s_3'',node distance=12mm}{anchor=north west,yshift=-1mm}{s_3'''}{$s_3'''$} \draw (s_2) to (s_3''');
      \btree{s_3'''}{\vec{t_3'''}}
    \end{scope}
  
    \begin{scope}[yshift=-50mm,xshift=0mm]
      \draw [very thick,->] (-5mm,22mm) to node [at end,right] {$\lput$} ++(0mm,-10mm);
      
      \anode[$\star$]{}{}{s_2}{$s_2$} \node [anchor=south east,xshift=-2mm,rotate=-55] at (s_2.north west) {$\dots$};
      \atree{s_2}{\vec{t_2}}
      \anode{below right of=s_2,xshift=-2mm}{}{s_3}{$s_3$} \draw (s_2) to (s_3);
      \atree{s_3}{\vec{t_3}}
      \anode{below right of=s_3,xshift=-2mm}{}{s_4}{$s_4$} \draw (s_3) to (s_4);
      \draw [ultra thick,red] ($(s_3)!.5!(s_4) + (-2mm,-1.5mm)$) to ++(4mm,3mm);
    \end{scope}
  
    \begin{scope}[yshift=-50mm,xshift=40mm]
      \draw [very thick,->] (-29mm,2mm) to node [at end,below,anchor=north west,rectangle,inner sep=.5mm,xshift=-1mm] {$\ldown$} ++(10mm,0mm);
  
      \anode{}{}{s_2}{$s_2$} \node [anchor=south east,xshift=-2mm,rotate=-55] at (s_2.north west) {$\dots$};
      \atree{s_2}{\vec{t_2}}
  
      \anode[$\star$]{below right of=s_2,xshift=-2mm}{anchor=north west,yshift=0mm}{s_3'}{$s_3'$} \draw (s_2) to (s_3');
      \btree{s_3'}{\vec{t_3'}}
      \anode[$\star$]{right of=s_3',node distance=20mm}{anchor=north west,yshift=0mm}{s_3''}{$s_3''$} \draw (s_2) to (s_3'');
      \btree{s_3''}{\vec{t_3''}}
      \anode[$\star$]{right of=s_3'',node distance=20mm}{anchor=north west,yshift=0mm}{s_3'''}{$s_3'''$} \draw (s_2) to (s_3''');
      \btree{s_3'''}{\vec{t_3'''}}
  
      \anode{below right of=s_3',xshift=-2mm}{anchor=north west}{s_4'}{$s_4'$} \draw (s_3') to (s_4');
      \draw [ultra thick,red] ($(s_3')!.6!(s_4') + (-2mm,-1.5mm)$) to ++(4mm,3mm);
      \anode{below right of=s_3'',xshift=-2mm}{anchor=north west}{s_4''}{$s_4''$} \draw (s_3'') to (s_4'');
      \draw [ultra thick,red] ($(s_3'')!.6!(s_4'') + (-2mm,-1.5mm)$) to ++(4mm,3mm);
      \anode{below right of=s_3''',xshift=-2mm}{anchor=north west}{s_4'''}{$s_4'''$} \draw (s_3''') to (s_4''');
      \draw [ultra thick,red] ($(s_3''')!.6!(s_4''') + (-2mm,-1.5mm)$) to ++(4mm,3mm);
  
      \draw [very thick,->] (15mm,2mm) to node [at end,below,anchor=north west,rectangle,inner sep=.5mm,xshift=-1mm,sloped] {$\ldown$} node [at end,above,anchor=south west,rectangle,inner sep=.5mm,xshift=-1mm,sloped] {$3$} ++(10mm,7mm);
    \end{scope}
  
  \end{tikzpicture}
  \caption{Simulating rule $\kpbody$ by star rules. All nodes of the trees have label 0.}\label{fig:kp:simulation}
\end{figure}%
}

There exist even more intimidating Hydras. A well-known one is the Buchholz Hydra, in which non-root nodes are labeled with ordinals $\alpha \leq \omega$ rather than $0$. For heads labeled with $\alpha = 0$, the Buchholz Hydra acts much like the Kirby-Paris Hydra. More interestingly, for heads labeled with $0 < \alpha < \omega$, the Hydra shifts down a (slightly modified)  copy of a supertree towards the chopping site. This allows it to not only grow in width, but also in height. The precise definition of the Buchholz Hydra contains a number of subtleties, given below.

\newcommand{\bh}{BH}
\newcommand{\bhto}{\to_\textit{\bh}}
\newcommand{\bhone}{\textsc{BH}_1}
\newcommand{\bhtwo}{\textsc{BH}_2}
\newcommand{\bhthree}{\textsc{BH}_3}
\begin{definition}[Buchholz Hydra~\cite{buchholz1987independence}]
A \emph{Buchholz Hydra} (\emph{\bh}) is a finite tree with nodes labeled from $L = \mathbb{N} \cup \{ \omega, \dagger \}$, where
\begin{enumerate}
    \item $\dagger$ labels only the root,
    \item\label{buch:hydra:below:root:zero} every node directly below the root is labeled $0$, and
    \item every remaining node is labeled with an ordinal $\alpha \leq \omega$.
\end{enumerate}

Its rewrite relation $\bhto$ can be defined as a TRS as induced by the following three rule schemas.

The first rule schema is a generalization of the Kirby-Paris rule schema $\kpbody$:
\begin{gather*}
    \begin{aligned}
(\bhone) && x(\vec{t_1}, \alpha(\vec{t_2}, \mathbf{0})) &
\to
x(\vec{t_1}, \underbrace{\alpha(\vec{t_2}), \ldots, \alpha(\vec{t_2})}_{\text{$k \in \mathbb{N}$ times}})
    \end{aligned}
  \end{gather*}
for any $x, \alpha \in L$.

The second rule is more intricate. For a targeted leaf node $x$ labeled with a successor ordinal $0 < \alpha < \omega$, find the closest ancestor $y$ such that $\mathit{label}(y) \leq \alpha$. (By condition \eqref{buch:hydra:below:root:zero} above, such an ancestor exists.)  Now copy the subtree $T$ rooted at $y$. In $T$, set $\mathit{label}(y) = \alpha - 1$, and set $\mathit{label}(x) = 0$. Call the result of these two changes $T'$. Then the second rule is given by
\begin{gather*}
    \begin{aligned}
      (\bhtwo) && z(\vec{t}, \mathbf{\alpha}) & \to z(\vec{t}, T')
    \end{aligned}
\end{gather*}
for any $z \in L$.

The third rule, finally, simply relabels a targeted leaf node with label $\omega$ to an arbitrary $k \in \mathbb{N}$, i.e.,
\begin{gather*}
    \begin{aligned}
      (\bhthree) && z(\vec{t}, \mathbf{\omega}) & \to z(\vec{t}, k) & \text{for $k \in \mathbb{N}$}
    \end{aligned}
\end{gather*}
for any $z \in L$.

\end{definition}

\begin{example}[Example application of $\bhtwo$]
Cutting off the boldfaced head in the Buchholz Hydra $\dagger(0(\omega), 0(2,7(\mathbf{5})))$ results in the Hydra
$\dagger(0(\omega), 0(2,7(T')))$ with $T' = 4(2,7(0))$.
\end{example}

Buchholz~\cite{buchholz1995proof} has shown that every Buchholz Hydra can be slain by always chopping the rightmost head. 

\begin{theorem}
  $\bhto$ is weakly normalizing.
\end{theorem}

We conjecture that $\bhto$ is even strongly normalizing. But even if this is correct, IPO cannot be used to prove this. IPO can simulate rules $\bhone$ and $\bhthree$, respectively, by the argument given in the proof to Theorem~\ref{thm:kp:can:be:slaim}, and the fact that IPO can decrease labels. But rule $\bhtwo$ cannot generally be simulated, as the proof to the following lemma demonstrates.

\begin{lemma}[IPO cannot simulate BH]\label{lem:ipo:cannot:simulate:bh}
${\bhto} \not\subseteq {\ilpored^*}$.
\end{lemma}
\begin{proof}
Assume by contradiction that ${\bhto} \subseteq {\ilpored^*}$. Then in particular, the $\bhtwo$ step
$\dagger(0(3(\mathbf{2}))) \bhto \dagger(0(3(1(3(0)))))$ implies the existence of a sequence
\[
\begin{array}{llll}
\dagger(0(3(2))) & \ilpored^* & \dagger(0(3(1(3(0))))) \\
& \ilpored_\lput \cdot \ilpored_\lselect  & \dagger(0(3(3(0)))) \\ 
& \ilpored_\lput & \dagger(0(3(3^\star(0)))) \\
& \ilpored_\lcopy  & \dagger(0(3(2))) \text{,} \\
\end{array}
\]
which is a $\ilpored^+$ cycle on unmarked trees, contradicting Lemma~\ref{lem:rpo:acyclic}.
\end{proof}

\begin{remark}
The proof to Lemma~\ref{lem:ipo:cannot:simulate:bh} shows something stronger: any pursuit for a conservative extension of IPO that allows one to simulate the Buchholz Hydra is fruitless. For any such solution results in cyclicity, and therefore non-termination, of $\ilpored^+$ on unmarked trees, defeating the purpose of defining $\ilpored$ in the first place.

The heart of the problem is that, under minimal conditions, the Buchholz Hydra allows replacing a leaf $x$ in a tree $T$ with a tree $T'$, where $T'$ contains a subtree $Y$ larger than $x$. Subsequently, subtree $Y$ can be installed on the position of $x$ using repeated applications of rule $\lselect$, producing a tree $T''$ larger than the starting tree $T$.
\end{remark}

\newcommand{\blood}[1]{\underline{#1}}

Note that the Kirby-Paris Hydra can grow only in width, and never in height, whereas the Buchholz Hydra can grow both in width and in height. We now define an alternative powerful Hydra which can be labeled with any well-ordered set, and which can grow both in width and in height, but which -- unlike the Buchholz Hydra -- can be proven terminating using IPO. The underlying idea is that a chop can be regarded as prefulfilling an ``obligation'' to make its supertrees smaller: hence the roots of the supertrees receive ``permission'' (denoted by underlining) to modify themselves in accordance with the rules of $\trsstar$.
For a labeled tree $t = f(\vec{u})$, we write $\blood{t}$ to denote 
the tree $\blood{f}(\vec{u})$ obtained from $t$ by underlining the root symbol.

\newcommand{\lchop}{\ilporule{chop}}
\newcommand{\lpropagate}{\ilporule{propagate}}
\newcommand{\lwiden}{\ilporule{widen}}
\newcommand{\llengthen}{\ilporule{lengthen}}
\newcommand{\lwaive}{\ilporule{waive}}
\begin{definition}[Star Hydra]\label{def:starhydra}
A \emph{Star Hydra} (\emph{SH}) is a finite tree with nodes labeled from a well-founded set $(S, {<})$.
For every $s \in S$ let $\blood{s}$ be a fresh symbol, and define $\blood{S} = \set{\blood{s} \mid s \in S}$. 
On finite trees labeled with $S \cup \blood{S}$ we define the relation $\to$ by the term rewrite system:
\[
\begin{array}{lrlll}
    (\lchop) & n(\alpha,\vec{x}) & \to & \blood{n}(\blood{\beta_1},\ldots,\blood{\beta_k},\vec{x}) & \text{$k \ge 0$, $\beta_1,\ldots,\beta_k < \alpha$} \\
    (\lpropagate) & n(\blood{m}(\vec{x}),\vec{y}) &\to & \blood{n}(\blood{m}(\vec{x}),\vec{y}) \\
    (\lwiden) &  \blood{n}(\blood{m}(\vec{x}),\vec{y}) &\to & \blood{n}(\underbrace{\blood{m}(\vec{x}),\ldots,\blood{m}(\vec{x})}_{\text{$0$ or more}},\vec{y}) \\
    (\llengthen) &  \blood{n}(\vec{x}) &\to & \blood{m}(\blood{n}(\vec{x})) & \text{$m < n$} \\
    (\lwaive) &  \blood{m}(\vec{x}) &\to & m(\vec{x})
\end{array}
\]
The Star Hydra rewrite relation $\stepstarhydra$ is defined as the
relation
\[
\to_\lchop \cdot \to_\lpropagate^! \cdot \ ({\to_\lwiden} \ \cup \ {\to_\llengthen})^* \ \cdot 
\to_\lwaive^!
\]
where $R^!$ denotes a maximal reduction of $R$, i.e., $x \mathrel{R}^! y$ if and only if $x \mathrel{R}^* y$ and there exists no $z$ such that $y \mathrel{R} z$.
\end{definition}

Intuitively, the Star Hydra can be described as follows. When Hercules chops off a head, the Hydra enrages, possibly regrowing any smaller number of heads at the wound (${\to_\lchop}$). In addition, it seeks to strengthen the neck from the body to the wound (${\to^!_\lpropagate}$). It does this by repeatedly lengthening and duplicating subtrees along the neck, in any order ($({\to_\lwiden} \ \cup \ {\to_\llengthen})^*$). Eventually it calms down and stabilizes (${\to^!_\lwaive}$), opening up another opportunity for Hercules to strike. An example $\stepstarhydra$ step is given in Figure~\ref{fig:star:hydra}.

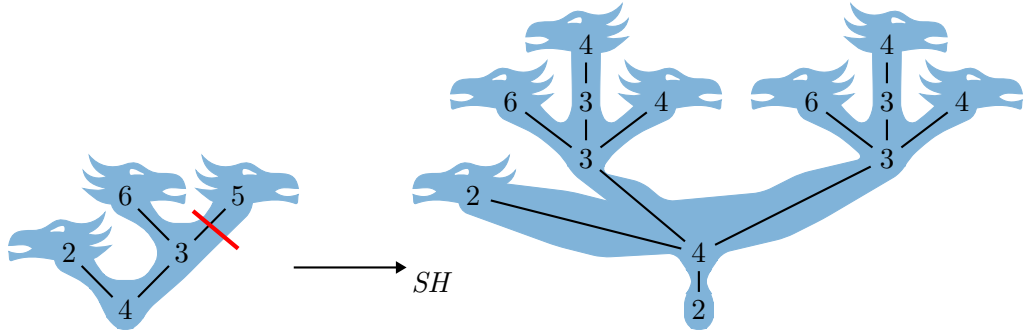
\begin{figure}
  \centering
  {
\newcommand{\x}[1]{\scalebox{1}[-1]{$#1$}}
\newcommand{\y}[1]{\x{\underline{#1}}}

\begin{tikzpicture}[]
  \node [anchor=south] at (0,0) {
    \scalebox{1}[-1]{%
    \begin{tikzpicture}[default,thick,level distance=10mm,nodes={circle}]
    
\node (x4) {\x{4}}
    child { node (x2) {\x{2}}
  }
    child { node (x3) {\x{3}}
      child { node (x6) {\x{6}}
    }
      child { node (x5) {\x{5}}
    }
  };
\begin{pgfonlayer}{background}
  \draw [rounded corners=1.5mm] [cblue!50,fill=cblue!50,rounded corners=1.2mm] ($(x4.east) + 0.5*(1mm,0mm)$) -- ($(x4.north east) + 0.5*(.4mm,.4mm)$) -- ($(x4.north) + 0.5*(0mm,1mm)$) -- ($(x4.north west) + 0.5*(-.4mm,.4mm)$) -- ($(x4.west) + 0.5*(-1mm,0mm)$) -- ($(x4) !0.2! 270:(x2) !.4! (x2)$) -- ($(x2) !0.2! 90:(x4) !.4! (x4)$) -- ($(x2.north west) + 0.5*(-.4mm,.4mm)$) -- ($(x2.south west) + 0.5*(-.4mm,-.4mm)$) -- ($(x2.south) + 0.5*(0mm,-1mm)$) -- ($(x2.south east) + 0.5*(.4mm,-.4mm)$) -- ($(x2) !0.2! 270:(x4) !.4! (x4)$) -- ($(x4) !0.2! 90:(x2) !.4! (x2)$) -- ($(x4) !0.2! 270:(x3) !.4! (x3)$) -- ($(x3) !0.2! 90:(x4) !.4! (x4)$) -- ($(x3.west) + 0.5*(-1mm,0mm)$) -- ($(x3) !0.2! 270:(x6) !.4! (x6)$) -- ($(x6) !0.2! 90:(x3) !.4! (x3)$) -- ($(x6.north west) + 0.5*(-.4mm,.4mm)$) -- ($(x6.south west) + 0.5*(-.4mm,-.4mm)$) -- ($(x6.south) + 0.5*(0mm,-1mm)$) -- ($(x6.south east) + 0.5*(.4mm,-.4mm)$) -- ($(x6) !0.2! 270:(x3) !.4! (x3)$) -- ($(x3) !0.2! 90:(x6) !.4! (x6)$) -- ($(x3) !0.2! 270:(x5) !.4! (x5)$) -- ($(x5) !0.2! 90:(x3) !.4! (x3)$) -- ($(x5.west) + 0.5*(-1mm,0mm)$) -- ($(x5.south west) + 0.5*(-.4mm,-.4mm)$) -- ($(x5.south) + 0.5*(0mm,-1mm)$) -- ($(x5.south east) + 0.5*(.4mm,-.4mm)$) -- ($(x5.north east) + 0.5*(.4mm,.4mm)$) -- ($(x5) !0.2! 270:(x4) !.4! (x4)$) -- ($(x4) !0.2! 90:(x5) !.4! (x5)$) -- cycle;
\end{pgfonlayer}

      \hydrahead{x2}{1}
      \hydrahead{x6}{-1}
      \hydrahead{x5}{-1}
 
     \draw [ultra thick,red] ($(x5)+(-6mm,2mm)$) to ++(6mm,5mm);
   \end{tikzpicture}
    }
  };

  \draw [->,thick] (18mm,10mm) to node [at end,anchor=north west,rectangle,inner sep=.5mm] {\textit{SH}} ++(15mm,0);

  \node [anchor=south] at (75mm,0mm) {
    \scalebox{1}[-1]{%
    \begin{tikzpicture}[default,thick,level distance=10mm,nodes={circle},level 2/.style={sibling distance=40mm},level 3/.style={sibling distance=10mm}]
    
\node (y2) {\x{2}}
    child { node (y4) {\x{4}}
      child { node[xshift=10mm] (2) {\x{2}}
    }
      child { node[xshift=-15mm,yshift=-5mm] (3a1) {\x{3}}
        child { node (6a) {\x{6}}
      }
        child { node (3a2) {\x{3}}
          child { node (4a1) {\x{4}}
        }
      }
        child { node (4a2) {\x{4}}
      }
    }
      child { node[xshift=-15mm,yshift=-5mm] (3b1) {\x{3}}
        child { node (6b) {\x{6}}
      }
        child { node (3b2) {\x{3}}
          child { node (4b1) {\x{4}}
        }
      }
        child { node (4b2) {\x{4}}
      }
    }
  };
\begin{pgfonlayer}{background}
  \draw [rounded corners=1.5mm] [cblue!50,fill=cblue!50,rounded corners=1.2mm] ($(y2.south east) + 0.5*(.4mm,-.4mm)$) -- ($(y2.north east) + 0.5*(.4mm,.4mm)$) -- ($(y2.north) + 0.5*(0mm,1mm)$) -- ($(y2.north west) + 0.5*(-.4mm,.4mm)$) -- ($(y2.south west) + 0.5*(-.4mm,-.4mm)$) -- ($(y2) !0.2! 270:(y4) !.4! (y4)$) -- ($(y4) !0.2! 90:(y2) !.4! (y2)$) -- ($(y4.north west) + 0.5*(-.4mm,.4mm)$) -- ($(y4.west) + 0.5*(-1mm,0mm)$) -- ($(y4) !0.2! 270:(2) !.4! (2)$) -- ($(2) !0.2! 90:(y4) !.4! (y4)$) -- ($(2.north west) + 0.5*(-.4mm,.4mm)$) -- ($(2.south west) + 0.5*(-.4mm,-.4mm)$) -- ($(2.south) + 0.5*(0mm,-1mm)$) -- ($(2.south east) + 0.5*(.4mm,-.4mm)$) -- ($(2) !0.2! 270:(y4) !.4! (y4)$) -- ($(y4) !0.2! 90:(2) !.4! (2)$) -- ($(y4) !0.2! 270:(3a1) !.4! (3a1)$) -- ($(3a1) !0.2! 90:(y4) !.4! (y4)$) -- ($(3a1.north west) + 0.5*(-.4mm,.4mm)$) -- ($(3a1.west) + 0.5*(-1mm,0mm)$) -- ($(3a1) !0.2! 270:(6a) !.4! (6a)$) -- ($(6a) !0.2! 90:(3a1) !.4! (3a1)$) -- ($(6a.north west) + 0.5*(-.4mm,.4mm)$) -- ($(6a.south west) + 0.5*(-.4mm,-.4mm)$) -- ($(6a.south) + 0.5*(0mm,-1mm)$) -- ($(6a.south east) + 0.5*(.4mm,-.4mm)$) -- ($(6a) !0.2! 270:(3a1) !.4! (3a1)$) -- ($(3a1) !0.2! 90:(6a) !.4! (6a)$) -- ($(3a1) !0.2! 270:(4a1) !.4! (4a1)$) -- ($(4a1) !0.2! 90:(3a1) !.4! (3a1)$) -- ($(4a1.north west) + 0.5*(-.4mm,.4mm)$) -- ($(4a1.south west) + 0.5*(-.4mm,-.4mm)$) -- ($(4a1.south) + 0.5*(0mm,-1mm)$) -- ($(4a1.south east) + 0.5*(.4mm,-.4mm)$) -- ($(4a1.north east) + 0.5*(.4mm,.4mm)$) -- ($(4a1) !0.2! 270:(3a1) !.4! (3a1)$) -- ($(3a1) !0.2! 90:(4a1) !.4! (4a1)$) -- ($(3a1) !0.2! 270:(4a2) !.4! (4a2)$) -- ($(4a2) !0.2! 90:(3a1) !.4! (3a1)$) -- ($(4a2.west) + 0.5*(-1mm,0mm)$) -- ($(4a2.south west) + 0.5*(-.4mm,-.4mm)$) -- ($(4a2.south) + 0.5*(0mm,-1mm)$) -- ($(4a2.south east) + 0.5*(.4mm,-.4mm)$) -- ($(4a2.north east) + 0.5*(.4mm,.4mm)$) -- ($(4a2) !0.2! 270:(3a1) !.4! (3a1)$) -- ($(3a1) !0.2! 90:(4a2) !.4! (4a2)$) -- ($(3a1.east) + 0.5*(1mm,0mm)$) -- ($(3a1.north east) + 0.5*(.4mm,.4mm)$) -- ($(3a1) !0.2! 270:(y4) !.4! (y4)$) -- ($(y4) !0.2! 90:(3a1) !.4! (3a1)$) -- ($(y4) !0.2! 270:(3b1) !.4! (3b1)$) -- ($(3b1) !0.2! 90:(y4) !.4! (y4)$) -- ($(3b1.west) + 0.5*(-1mm,0mm)$) -- ($(3b1) !0.2! 270:(6b) !.4! (6b)$) -- ($(6b) !0.2! 90:(3b1) !.4! (3b1)$) -- ($(6b.north west) + 0.5*(-.4mm,.4mm)$) -- ($(6b.south west) + 0.5*(-.4mm,-.4mm)$) -- ($(6b.south) + 0.5*(0mm,-1mm)$) -- ($(6b.south east) + 0.5*(.4mm,-.4mm)$) -- ($(6b) !0.2! 270:(3b1) !.4! (3b1)$) -- ($(3b1) !0.2! 90:(6b) !.4! (6b)$) -- ($(3b1) !0.2! 270:(4b1) !.4! (4b1)$) -- ($(4b1) !0.2! 90:(3b1) !.4! (3b1)$) -- ($(4b1.north west) + 0.5*(-.4mm,.4mm)$) -- ($(4b1.south west) + 0.5*(-.4mm,-.4mm)$) -- ($(4b1.south) + 0.5*(0mm,-1mm)$) -- ($(4b1.south east) + 0.5*(.4mm,-.4mm)$) -- ($(4b1.north east) + 0.5*(.4mm,.4mm)$) -- ($(4b1) !0.2! 270:(3b1) !.4! (3b1)$) -- ($(3b1) !0.2! 90:(4b1) !.4! (4b1)$) -- ($(3b1) !0.2! 270:(4b2) !.4! (4b2)$) -- ($(4b2) !0.2! 90:(3b1) !.4! (3b1)$) -- ($(4b2.west) + 0.5*(-1mm,0mm)$) -- ($(4b2.south west) + 0.5*(-.4mm,-.4mm)$) -- ($(4b2.south) + 0.5*(0mm,-1mm)$) -- ($(4b2.south east) + 0.5*(.4mm,-.4mm)$) -- ($(4b2.north east) + 0.5*(.4mm,.4mm)$) -- ($(4b2) !0.2! 270:(3b1) !.4! (3b1)$) -- ($(3b1) !0.2! 90:(4b2) !.4! (4b2)$) -- ($(3b1.east) + 0.5*(1mm,0mm)$) -- ($(3b1.north east) + 0.5*(.4mm,.4mm)$) -- ($(3b1) !0.2! 270:(y4) !.4! (y4)$) -- ($(y4) !0.2! 90:(3b1) !.4! (3b1)$) -- ($(y4.east) + 0.5*(1mm,0mm)$) -- ($(y4.north east) + 0.5*(.4mm,.4mm)$) -- ($(y4) !0.2! 270:(y2) !.4! (y2)$) -- ($(y2) !0.2! 90:(y4) !.4! (y4)$) -- ($(y2.south east) + 0.5*(.4mm,-.4mm)$) -- ($(y2.north east) + 0.5*(.4mm,.4mm)$) -- ($(y2) !0.2! 270:(y2) !.4! (y2)$) -- ($(y2) !0.2! 90:(y2) !.4! (y2)$) -- cycle;
\end{pgfonlayer}

      \hydrahead{2}{1}
      \hydrahead{6a}{1}
      \hydrahead{4a1}{1}
      \hydrahead{4a2}{-1}
      \hydrahead{6b}{1}
      \hydrahead{4b1}{-1}
      \hydrahead{4b2}{-1}
    \end{tikzpicture}
    }
  };

\end{tikzpicture}
}
  \caption{One step in the process of killing the Star Hydra.
    After chopping the head with value $5$, the Hydra grows both in width and in height,
    and Hercules starts running.}\label{fig:star:hydra}
  \label{fig:starhydra}
\end{figure}

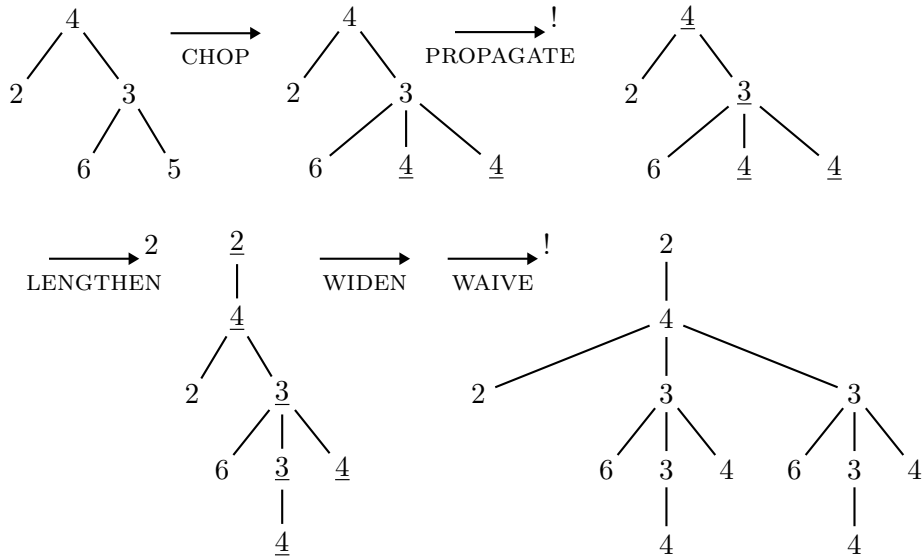
\begin{figure}
  \centering
  {
\newcommand{\x}[1]{\scalebox{1}[1]{$#1$}}
\newcommand{\y}[1]{\x{\underline{#1}}}

\tikzset{hydra/.style={default,thick,level distance=10mm,nodes={circle},level 2/.style={sibling distance=12mm},level 3/.style={sibling distance=8mm}}}

\begin{tikzpicture}[]
  \node at (0,0) {
    \scalebox{1}[1]{%
    \begin{tikzpicture}[hydra]
    
\node (x4) {\x{4}}
    child { node (x2) {\x{2}}
  }
    child { node (x3) {\x{3}}
      child { node (x6) {\x{6}}
    }
      child { node (x5) {\x{5}}
    }
  };
\begin{pgfonlayer}{background}
\end{pgfonlayer}

      \draw [->,thick] (13mm,-2mm) to node [anchor=north,rectangle,inner sep=2mm] {$\lchop$} ++(12mm,0mm);
    \end{tikzpicture}
    }
  };
  \node at (40mm,0) {
    \scalebox{1}[1]{%
    \begin{tikzpicture}[hydra]
    
\node (x4) {\x{4}}
    child { node (x2) {\x{2}}
  }
    child { node (x3) {\x{3}}
      child { node (x6) {\x{6}}
    }
      child { node (4a) {\y{4}}
    }
      child { node (4b) {\y{4}}
    }
  };
\begin{pgfonlayer}{background}
\end{pgfonlayer}

      \draw [->,thick] (14mm,-2mm) to node [anchor=north,rectangle,inner sep=2mm] {$\lpropagate$} node [at end,anchor=south west,rectangle,inner sep=0.5mm] {$!$} ++(12mm,0mm);    
    \end{tikzpicture}
    }
  };
  \node at (80mm,0) {
    \scalebox{1}[1]{%
    \begin{tikzpicture}[hydra]
    
\node (y4) {\y{4}}
    child { node (x2) {\x{2}}
  }
    child { node (y3) {\y{3}}
      child { node (x6) {\x{6}}
    }
      child { node (4a) {\y{4}}
    }
      child { node (4b) {\y{4}}
    }
  };
\begin{pgfonlayer}{background}
\end{pgfonlayer}

    \end{tikzpicture}
    }
  };
  \node at (20mm,-40mm) {
    \scalebox{1}[1]{%
    \begin{tikzpicture}[hydra]
    
\node (y2) {\y{2}}
    child { node (y4) {\y{4}}
      child { node (x2) {\x{2}}
    }
      child { node (y3) {\y{3}}
        child { node (x6) {\x{6}}
      }
        child { node {\y{3}}
          child { node {\y{4}}
        }
      }
        child { node {\y{4}}
      }
    }
  };
\begin{pgfonlayer}{background}
\end{pgfonlayer}

      \draw [->,thick] (-25mm,-2mm) to node [anchor=north,rectangle,inner sep=2mm] {$\llengthen$} node [at end,anchor=south west,rectangle,inner sep=0.5mm] {$2$} ++(12mm,0mm);    
      \draw [->,thick] (11mm,-2mm) to node [anchor=north,rectangle,inner sep=2mm] {$\lwiden$} ++(12mm,0mm);    
      \draw [->,thick] (28mm,-2mm) to node [anchor=north,rectangle,inner sep=2mm] {$\lwaive$} node [at end,anchor=south west,rectangle,inner sep=0.5mm] {$!$} ++(12mm,0mm);    
    \end{tikzpicture}
    }
  };
  
  \node at (75mm,-40mm) {
    \scalebox{1}[1]{%
    \begin{tikzpicture}[hydra,level 2/.style={sibling distance=25mm}]
    
\node (y2) {\x{2}}
    child { node (y4) {\x{4}}
      child { node (2) {\x{2}}
    }
      child { node (3a1) {\x{3}}
        child { node (6a) {\x{6}}
      }
        child { node (3a2) {\x{3}}
          child { node (4a1) {\x{4}}
        }
      }
        child { node (4a2) {\x{4}}
      }
    }
      child { node (3b1) {\x{3}}
        child { node (6b) {\x{6}}
      }
        child { node (3b2) {\x{3}}
          child { node (4b1) {\x{4}}
        }
      }
        child { node (4b2) {\x{4}}
      }
    }
  };
\begin{pgfonlayer}{background}
\end{pgfonlayer}

    \end{tikzpicture}
    }
  };
\end{tikzpicture}
}
  \caption{Justification of the $\stepstarhydra$ step in Figure~\ref{fig:starhydra},
    detailing the change of the Hydra via the steps in Definition~\ref{def:starhydra}.
  }
  \label{fig:starhydra:details}
\end{figure}

\newcommand{\remtrash}{\hookrightarrow}
\newcommand{\dropul}[1]{[#1]}

In order to prove termination of $\stepstarhydra$ we show that ${\stepstarhydra} \subseteq {\ilpored^+}$.
For this purpose, we show that the underlined reduction 
introduced in Definition~\ref{def:starhydra},
in particular the rules $\lwiden$ and $\llengthen$, can be simulated by $\ilpored^+$.
As a first approximation, these rules can be simulated by
\begin{align*}
    n(m(\vec{x}),\vec{u})
      &\ilpored_{\lput} n^\star(m(\vec{x}),\vec{u}) \\ 
      &\ilpored_{\ldown} n(m^\star(\vec{x}),\ldots,m^\star(\vec{x}),\vec{u})
\end{align*}
and
\begin{align*}
    n(\vec{x})
      &\ilpored_{\lput} n^\star(t,\vec{x}) \\ 
      &\ilpored_{\lcopy} m(n^\star(\vec{x})) \;,
\end{align*}
respectively.
However, there is one problem: we need to get rid of the stars $\star$ 
after the $\ilpored$ simulation.
The idea is to enrich the trees such that every underlined symbol has additional ``garbage'' arguments.
A garbage argument $\textit{garbage}$ can then be used to eliminate a $\star$ symbol as follows:
\[
  n^\star(\vec{t},\textit{garbage}) \ilpored_{\ldown} n(\vec{t})
\]
To this end, we introduce the notion of garbage collection.

For convenience, when proving ${\stepstarhydra} \subseteq {\ilpored^+}$, we consider $\ilpored$ over an extended signature
\[ S' = S \cup \set{\bot} \]
where $\bot$ is a fresh element that is smaller than all elements in $S$. So $\bot$ is an auxiliary element that is not used in the Hydra reduction $\stepstarhydra$.

\begin{definition}[Garbage collection]
  For every $g \ge 1$ we define $\remtrash_g$ on $\tre{S' \cup \blood{S'}}{\avars}$ by the following rewrite system
  \begin{align*}
    n(\vec{x},y_1,\ldots,y_h) \remtrash_g \blood{n}(\vec{x}) \quad \text{for every $h \ge g$}
  \end{align*}
\end{definition}

\begin{definition}[Dropping underlining]
  For $s \in \tre{S' \cup \blood{S'}}{\avars}$ we write
  $\dropul{s}$ for the tree obtained from $s$ by removing underlining,
  that is, replacing each symbol occurrence $\blood{n}$ by $n$.
\end{definition}

Garbage collection can be handled by $\ilpored$ in the following sense.

\begin{lemma}\label{lem:garbage:ilpo}
  If $s \in \tre{S'}{\avars}$ and $s \remtrash_g^* t$,
  then $s \ilpored^* \dropul{t}$.
\end{lemma}

\begin{proof}
  Since $s \remtrash_g^* t$, the term $\dropul{t}$ is obtained from $s$ by removing at least $g$ subtrees.
  This can be simulated by repeated applications of $\ilpored_{\lput} \cdot \ilpored_{\ldown}$.
\end{proof}

The following lemma states that $\ilpored^*$ can simulate the chopping of a head
and insert garbage arguments along the affected neck (the underlined symbol occurrence).

\begin{lemma}\label{lem:garbage:chop}
  If $g \geq 1$, $s \in \tre{S}{\avars}$ and $s \to_\lchop \cdot \to_\lpropagate^! t$,
  then $s \ilpored s^\star \ilpored^* v \remtrash_g^* t$ for some $v \in \tre{S'}{\avars}$.
\end{lemma}

\begin{proof}
  We prove the claim by induction on the depth of the $\lchop$ step:
  \begin{enumerate}
    \item 
      If the $\lchop$ step occurs at the root, then it is of the form 
      \[ s = n(\alpha,\vec{u}) \to_\lchop \blood{n}(\blood{\beta_1},\ldots,\blood{\beta_k},\vec{u}) = t \]
      with $k \ge 0$, $\beta_1,\ldots,\beta_k < \alpha$. We then reduce as follows:
      \[
      \begin{array}{rll}
        s = n(\alpha,\vec{u}) 
        & \ilpored_{\lput} & n^\star(\alpha,\vec{u}) = s^\star \\
        &\ilpored_{\ldown}& n(\underbrace{\alpha^\star,\ldots,\alpha^\star}_{\text{$k+g$ times}},\vec{u}) \\
        & \ilpored_{\lcopy}^* & n(\beta_1(\underbrace{\bot,\ldots,\bot}_{\text{$g$ times}}),\ldots,\beta_k(\underbrace{\bot,\ldots,\bot}_{\text{$g$ times}}),\underbrace{\bot,\ldots,\bot}_{\text{$g$ times}},\vec{u}) = v\\
        & \remtrash_{g}^*& \blood{n}(\blood{\beta_1},\ldots,\blood{\beta_k},\vec{u}) = t
      \end{array}
      \]
    \item 
      Let $s = n(s',\vec{u})$ and assume that the $\lchop$ step occurs in $s'$.
      Then
      \[
        s' \to_\lchop \cdot \to_\lpropagate^! t'
        \quad \text{and} \quad
        n(t',\vec{u}) \to_\lpropagate \blood{n}(t',\vec{u}) = t
      \]
      for some $t'$.
      By the induction hypothesis $s' \ilpored s'^\star \ilpored^* v' \remtrash_{g}^* t'$ for some $v' \in \tre{S'}{\avars}$. Hence
      \[
      \begin{array}{rll}
        n(s',\vec{u}) 
        & \ilpored_{\lput} & n^\star(s',\vec{u}) \\
        &\ilpored_{\ldown}& n(\underbrace{s'^\star,\ldots,s'^\star}_{\text{$g+1$ times}},\vec{u}) \\
        &\ilpored^*&  n(\underbrace{v',\ldots,v'}_{\text{$g+1$ times}},\vec{u}) = v\\
        &\remtrash_{g}^*& \blood{n}(t',\vec{u}) = t
      \end{array}
      \]
  \end{enumerate}
  This concludes the proof.
\end{proof} 

\begin{lemma}\label{lem:garbage:widen}
  If $s \in \tre{S'}{\avars}$ and $s \remtrash_{g+1}^* \cdot \to_{\lwiden} u$ for $g \ge 1$,
  then $s \ilpored^* v \remtrash_{g}^* u$ for some $v \in \tre{S'}{\avars}$.
\end{lemma}

\begin{proof}
  Let $s \in \tre{S'}{\avars}$ and $s \remtrash_{g+1}^* t \to_{\lwiden} u$ with $g \ge 1$.
  W.l.o.g.\ we may assume that the $\lwiden$ step occurs at the  root
  (otherwise we consider the relevant subtrees of $s,t,u$).
  Then the rewrite sequence is of the form
  \begin{align*}
     s \remtrash_{g+1}^* t = \blood{n}(\blood{m}(\vec{v}),\vec{w}) 
    &\to_{\lwiden} \blood{n}(\blood{m}(\vec{v}),\ldots,\blood{m}(\vec{v}),\vec{w}) = u
  \end{align*}
  for some $n,m \in S'$ and $\vec{v},\vec{w} \in \tre{S' \cup \blood{S'}}{\avars}^*$.
  As a consequence $s$ must be of the form
  \begin{align*}
     s = n(m(\vec{v'},x_1,\ldots,x_{g_1+1}),\vec{w'},y_1,\ldots,y_{g_2+1}) 
  \end{align*}
  with $g_1,g_2 \ge g$
  and componentwise $\vec{v'} \remtrash_{g+1}^* \vec{v}$ and $\vec{w'} \remtrash_{g+1}^* \vec{w}$.
  Then
  \[
  \begin{array}{rll}
     s &=& n(m(\vec{v'},x_1,\ldots,x_{g_1+1}),\vec{w'},\vec{y}) \\
     &\ilpored_{\lput}& n^\star(m(\vec{v'},x_1,\ldots,x_{g_1+1}),\vec{w'},\vec{y}) \\
     &\ilpored_{\ldown}& n(m^\star(\vec{v'},x_1,\ldots,x_{g_1+1}),\ldots,m^\star(\vec{v'},x_1,\ldots,x_{g_1+1}),\vec{w'},\vec{y}) \\ 
     &\ilpored_{\ldown}^*& n(m(\vec{v'},x_1,\ldots,x_{g_1}),\ldots,m(\vec{v'},x_1,\ldots,x_{g_1}),\vec{w'},\vec{y}) = v \\
     &\remtrash_{g}^*& \blood{n}(\blood{m}(\vec{v}),\ldots,\blood{m}(\vec{v}),\vec{w}) \\
     &=& u
  \end{array}
  \]
  This proves the claim.
\end{proof}

\begin{lemma}\label{lem:garbage:lengthen}
  If $s \in \tre{S'}{\avars}$ and $s \remtrash_{g+1}^* \cdot \to_{\llengthen} u$ for $g \ge 1$,
  then $s \ilpored^* v \remtrash_{g}^* u$ for some $v \in \tre{S'}{\avars}$.
\end{lemma}

\begin{proof}
  Let $s \in \tre{S'}{\avars}$ and $s \remtrash_{g+1}^* t \to_{\llengthen} u$ with $g \ge 1$.
  Assume that the $\llengthen$ step occurs at the  root
  (otherwise we consider the relevant subtrees of $s,t,u$).
  The rewrite sequence is of the form
  \begin{align*}
     s \remtrash_{g+1}^* t = \blood{n}(\vec{w}) 
    &\to_{\llengthen} \blood{m}(\blood{n}(\vec{w})) = u
  \end{align*}
  for some $n > m \in S'$ and $\vec{w} \in \tre{S' \cup \blood{S'}}{\avars}^*$.
  Hence $s$ must be of the form
  \begin{align*}
     s = n(\vec{w'},x_1,\ldots,x_{g'+1})
  \end{align*}
  with $g' \ge g$
  and componentwise $\vec{w'} \remtrash_{g+1}^* \vec{w}$.
  Then
  \[
  \begin{array}{rll}
     s & = & n(\vec{w'},x_1,\ldots,x_{g'+1}) \\
     &\ilpored_{\lput} &  n^\star(\vec{w'},x_1,\ldots,x_{g'+1}) \\
     &\ilpored_{\lcopy}& m(\underbrace{n^\star(\vec{w'},x_1,\ldots,x_{g'+1}),\ldots,n^\star(\vec{w'},x_1,\ldots,x_{g'+1})}_{\text{$g+1$ times}} ) \\ 
     &\ilpored_{\ldown}^*& m(n(\vec{w'},x_1,\ldots,x_{g'}),\ldots,n(\vec{w'},x_1,\ldots,x_{g'}) ) \\
     &\remtrash_{g}^*&  \blood{m}( \blood{n}(\vec{w}) ) \\
     &=& u
    \end{array}
    \]
  This concludes the proof.
\end{proof}

\begin{lemma}\label{lem:starhydra:ilpo}
  We have ${\stepstarhydra} \subseteq {\ilpored^+}$ where $\ilpored$ uses the given order $<$ on $S'$.
\end{lemma}

\begin{proof}
  Assume that $s \stepstarhydra t$. Then
  \[
  s \to_\lchop \cdot \to_\lpropagate^! s'' \mathrel{ ({\to_\lwiden} \ \cup \ {\to_\llengthen})^n } t' \to_\lwaive^! t
  \]
  for some $n \ge 0$ and trees $s'',t'$.
  By Lemma~\ref{lem:garbage:chop} we have $s \ilpored^+ s' \remtrash_{n+1}^* s''$ for some $s' \in \tre{S'}{\avars}$.
  Using induction on $n$ together with Lemmas~\ref{lem:garbage:widen} and~\ref{lem:garbage:lengthen},
  we obtain $s' \ilpored^* t'' \remtrash_1 t'$ for some $t'' \in \tre{S'}{\avars}$.
  Thus we have
  \[
    s \ilpored^* s' \ilpored^* t'' \remtrash_1 t' \to_\lwaive^! t 
  \]
  Moreover, by Lemma~\ref{lem:garbage:ilpo} we have $t'' \ilpored^* \dropul{t'} = t$.
  Hence $s \ilpored^* t$.
\end{proof}

\begin{theorem}\label{thm:starhydra}
  The Star Hydra $\stepstarhydra$ is terminating.
\end{theorem}

\begin{proof}
  A direct consequence of Lemma~\ref{lem:starhydra:ilpo}.
\end{proof}

\section{Concluding remarks and questions}\label{sec:conclusion}

We have revisited star games as an alternative setting for dealing with several termination problems, to wit, of word rewriting systems, term rewriting systems, and other objects, in this paper in particular Hydras. We expect that the scope of this method, consisting of letting some control symbols walk over the objects and transforming them, is wider than we could treat in this paper. We mention a few follow-up matters, that we hope to address in subsequent notes or papers.

\begin{enumerate}
\item An essential part of the star games is that they cooperate with well-quasi-orders, as vocalized in our simple but basic Proposition~\ref{basic}. That requires that the appropriate embedding relations in the wqo can be obtained by a subset of the star game rules, or by derived rules, derived from the star rules. Above we have seen this for the notion of embedding as in Kruskal's tree theorem.

A  challenge is to extend the star rules so that they capture the more difficult embedding on words known as \emph{gap embedding}, and likewise the corresponding embedding on trees as proposed by Friedman
and studied by Dershowitz and Tzameret~\cite{dershowitz2003gap,tzameret2002kruskal}. %

\item In general, it is our endeavour to extend the \emph{expressivity} of the star rules by adding new rules. A typical example, already realized in the literature, is admitting the \llex{} rule as treated in the paper ILPO~\cite{klop2006iterative}, that actually makes the $\ldown$ rule a derived rule, and that importantly, enables one to prove termination of functions like the one of Ackermann. See further Geser~\cite{gese:1990}. However, in the setting of trees (as we have used in this paper), the \llex{} rule is unsound for proving termination (see Remark~\ref{rem:lex}).

\item In continuation of the previous item about enhancing expressivity,
fruitful results have already been obtained towards connecting the framework of this paper
to \emph{higher-order rewriting}, pertaining to various typed and untyped lambda calculi, or combinations 
of these with first-order rewriting such as in HRSs and CRSs \cite{klop1993term}.

For the original version of RPO, progress in this direction has been made by Jouannaud and Rubio and others. 
For the equivalent version of IPO as in~\cite{klop2006iterative}  and the present paper, several developments have been realized by Kop~\cite{kop2012higher} and Van Raamsdonk~\cite{kop2009iterative}.

\item
Another direction is the adaptation of the star game for proving infinitary normalization~\cite{infinitary:normalization:2009} and productivity~\cite{endr:2010,complexity:productivity:2009}.

\item
We also would like  to enhance the \emph{precision} of the star rules. Can the star game be changed in order to prove \emph{local termination}~\cite{termination:local:2009,termination:local:2010} on a given set of starting terms, in contrast to \emph{global termination} on all terms?

\item The star game framework lends itself to a further syntactical analysis
of its fine-structure along the lines of obtaining for instance a \emph{standardization theorem}, well-known in lambda calculus and first order rewriting.
Such a theorem was already 
included and used in~\cite{klop2006iterative}, but as yet without proof, which was subsequently given by Van Oostrom under the name of 
obtaining \emph{cut elimination}. An interesting corollary pointed out by him in that work is the \emph{decidability} of the star rewrite relation.
The point of such a strategy is that computation is going down along the tree with a \emph{wave front} of star symbols, denoting the front of activity, and freezing everything above that front.

\item 
Are there interesting variations of star games for proving termination in infinitary rewriting~\cite{infinitary:highlights:2012,infintary:rewriting:coinductive:2018}, that is, considering infinite trees and rewrite sequences of ordinal length?

\item A further question that we wondered about, is obtained by \emph{reflecting a star game in itself}, and continuing to repeat that.  That is, we can allow finite $\omega$-labeled trees themselves as labels, and determine the termination ordinal thus obtained. We can reiterate this feeding the system into itself, even in an ordinal long procedure. Here the notion of normal function on ordinals will likely play a role. Can we determine at what ordinal this process closes? We expect that for some of these questions there are already answers in the literature, but we have not encountered such answers explicitly, for several questions.

\item The most fascinating question to our mind concerns the extension of star games to \emph{finite graphs}, to start with, undirected  and unlabeled as in the original \emph{graph minor theorem (GMT)} by Robertson and Seymour~\cite{robertson2004graph}.
Obvious generalizations here are admitting directed edges, and admitting ordinal labels on the nodes. This is a very powerful generalization of KTT. Ideally, we would like to obtain a star game on such finite graphs, yielding a framework to termination of various graph rewrite problems, by simulation into this terminating graph framework~\cite{joshi2008applying}. A promising advance in this direction has been  made by Dershowitz and Jouannaud, who designed a very general datatype called \emph{drags}, not using star games, but with developments towards setting up RPO, recursive path order, for these objects~\cite{dershowitz2018drags,dershowitz2018graph}. 
\end{enumerate}

\section*{Acknowledgments}
We thank Nachum Dershowitz for his useful conversations concerning a draft of this paper and 
many comments and pointers to the literature, during his sabbatical period at CWI Amsterdam around end of 2019 and beginning of 2020.
Also we gratefully acknowledge support for our research from CWI Amsterdam and the Section Theory of the Vrije Universiteit Amsterdam. The third author received funding from NWO under the Innovational Research Incentives Scheme (project No.\ VI.Vidi.192.004). The second author is thankful for monthly funding from the Dutch ABP, the Algemeen Burgerlijk Pensioenfonds, for daily subsistence and cost of living.

Many thanks to Vincent van Oostrom, for a scrutiny of a draft of this paper and numerous comments for improvement. We thank Andreas Weiermann and Toshiyasu Arai for advising us on references for ordinal analysis. We finally thank two anonymous referees, for several more comments and valuable suggestions.

\bibliographystyle{alpha}
\bibliography{main}

\end{document}